%% file: p3.tex
\documentclass[a4paper,reqno]{amsart}

\usepackage{amsmath,amssymb,amsfonts}
\usepackage{graphicx}
\usepackage{changepage}
\usepackage[alphabetic,y2k,initials]{amsrefs}
\usepackage{pst-all}
\usepackage{pst-solides3d}
\usepackage{longtable}

\let\oldtocsection=\tocsection
 \let\oldtocsubsection=\tocsubsection
 \let\oldtocsubsubsection=\tocsubsubsection
 
\renewcommand{\tocsection}[2]{\hspace{1em}\oldtocsection{#1}{#2}}
\renewcommand{\tocsubsection}[2]{\hspace{3em}\oldtocsubsection{#1}{#2}}
\renewcommand{\tocsubsubsection}[2]{\hspace{3.5em}\oldtocsubsubsection{#1}{#2}}

\def\const{\mathrm{const}}

\def\arctan{\mathrm{arctan}}
\def\Pic{\mathrm{Pic}}

\def\PIII{\mathrm{P}_{\mathrm{III}}}
\def\PV{\mathrm{P}_{\mathrm{V}}}

\numberwithin{equation}{section}

\newtheorem{theorem}{Theorem}[section]
\newtheorem{proposition}[theorem]{Proposition}
\newtheorem{corollary}[theorem]{Corollary}
\newtheorem{lemma}[theorem]{Lemma}
\newtheorem{remark}[theorem]{Remark}

\newtheorem{definition}[theorem]{Definition}

\definecolor{lightgray}{gray}{0.9}

\colorlet{crvena}{black!80}
\colorlet{plava}{black!50}
\colorlet{zelena}{black!20}

\date{}

\begin{document}
\title[Asymptotics of $\PIII$]{Asymptotic behaviour of the third Painlev\'e transcendents in the space of initial values}
\author{Nalini Joshi}
\thanks{This research was supported by an Australian Laureate Fellowship \# FL120100094 from the Australian Research Council.
The research of M.R.~is partially supported by  the Serbian Ministry of Education and Science (Project no. 174020: Geometry and Topology of Manifolds and Integrable Dynamical Systems).
}
\address{School of Mathematics and Statistics F07, University of Sydney, NSW 2006, Australia}
\email{nalini.joshi@sydney.edu.au}

\author{Milena Radnovi\'c}
\address{School of Mathematics and Statistics F07, University of Sydney, NSW 2006, Australia}
\email{milena.radnovic@sydney.edu.au}

\begin{abstract}
We study the asymptotic behaviour of the solutions of the generic
($D_6^{(1)}$-type) third Painlev\'e equation in the space of initial
values as the independent
variable approaches infinity (or zero) 
and show that the limit set of each solution is compact and
connected. Moreover, we prove that any solution with essential
singularity at infinity has an infinite number of poles and zeroes,
and similarly at the origin. 
\end{abstract}

\maketitle

\tableofcontents

\section{Introduction}
\input{1-intro/intro}

\section{Space of initial values of $\PIII$}\label{s:ivs}
\input{2-space/space}

\subsection{A system equivalent to $\PIII$}\label{sec:system}
\input{2-space/system}

\subsection{Autonomous equation}\label{sec:auto}
\input{2-space/auto}

\subsection{Resolution of singularities}\label{sec:okamoto}
\input{2-space/okamoto}

\subsection{Movable singularities}\label{sec:movable}
\input{2-space/movable}

\section{Special solutions of $\PIII$}\label{s:ssol}
\input{3-special/special}

\subsection{Special rational solutions of $\PIII$}
\input{3-special/rational}

\section{Behaviour near the infinity set}\label{s:inf}
\input{4-infinity/infinity}

\section{The limit set}\label{s:lim}
\input{5-limit/limit}

\appendix


\section{Resolution of the Painlev\'e vector field}\label{sec:p3resolution}
\input{A-p3resolution/p3resolution}

\section{Notation}\label{sec:notation}
\input{B-notation/notation}


\end{document}

%% file: 1-intro/intro.tex
Following interest in transcendental
solutions of the third Painlev\'e equation 
$$
\PIII :\quad  y''=\frac{y'^2}{y}-\frac{y'}{x}+\frac{\alpha y^2+\beta}{x}+\gamma y^3+\frac{\delta}{y},
$$
due to physical applications (see
e.g. \cites{McCoyTracyWu1977,kitaev1989,Kanzieper02}),
we describe and prove global
properties of such solutions by studying their asymptotic behaviours in the limit $x\to\infty$
(equivalently $x\to0$) in the initial value space. In this paper, we find complete information about the limit sets of transcendental solutions. Unlike previous studies, we do not impose any reality constraints on the solutions, nor assume special conditions on the parameters; $y$ is a function of the complex variable $x$ and $\alpha$, $\beta$, $\gamma$, $\delta$ are given complex constants with $\gamma\delta\not=0$.

Our main results are obtained after regularization and compactification of
the space of initial values of
$\PIII$ first described by Okamoto
\cite{Okamoto1979}. A detailed description of this space is given in Section \ref{sec:okamoto} below. The space needs to be regularized by using resolution (or blow-ups) due to problematic points, called base points, where the Painlev\'e vector field becomes undefined (for example, it approaches $0/0$). We carry out resolutions for $x\to\infty$ and use asymptotic estimates to deduce previously unknown properties of solutions in the limit.  

Successive resolutions of
the vector field at base points terminates after nine blow ups of
$\mathbb C\mathbb P^2$ in the cases of the first, second, and fourth Painlev\'e equations \cites{DJ2011,HJ2014,JR2016}, 
while for the fifth and third Painlev\'e equations the construction consists of eleven blow ups and two blow downs \cite{JR2017}.
The initial value space in each case is then obtained by removing the infinity set, denoted $\mathcal{I}$, which are blow-ups of points not reached by any solution.

The generic case of $\PIII$, which we study in this paper, is a degenerate case of the fifth Painlev\'e equation
$\PV$. The generic case of $\PV$ was analysed recently  in \cite{JR2017} in its initial value space. Although no further
advances are required to tackle the special case of $\PIII$ when $\gamma=0$
or $\delta=0$, the construction becomes more technical because base points merge in the construction of the initial
value space, and we do not include this case in the present paper. 

Our main results fall into three parts:
\vspace{2pt}
\begin{list}{}
  {\usecounter{enumi}
    \setlength{\parsep}{2pt}
    \setlength{\leftmargin}{12pt}\setlength{\rightmargin}{12pt}
    \setlength{\itemindent=-12pt}
  }

\item {\em Existence of a repeller set:} Theorem \ref{th:estimates} in Section \ref{s:inf} shows that $\mathcal I$ is a repeller for the flow. The theorem also provides the range of the independent variable for which a solution may remain in the vicinity of $\mathcal{I}$.
\item {\em Numbers of poles and zeroes:} In Corollary \ref{cor:infinity}, we prove that each solution that is sufficiently close to $\mathcal{I}$ has a pole in a neighbourhood of the corresponding value of the independent variable. Moreover, Theorem \ref{th:poleszeroes} shows that each solution with essential singularity at $x=\infty$ has infinitely many poles and infinitely many zeroes in each neighbourhood of that point.
\item {\em The complex limit set:} We prove in Theorem \ref{th:limit} that the limit set for each solution is non-empty, compact, connected, and invariant under the flow of the autonomous equation obtained as $x\to\infty$.
\end{list}
\vspace{2pt}

This paper is organised as follows.
In Section \ref{s:ivs}, we construct and describe the space of the initial values for equation $\PIII$, with
complete details of all the necessary calculations provided in Appendix \ref{sec:p3resolution}.
We also relate the trajectories crossing exceptional lines to movable singularities of solutions
(both poles and zeroes) and discuss the autonomous
system and its Hamiltonian. In Section \ref{s:ssol}, we consider the special solutions of $\PIII$ and relate them to the fixed points of the Hamiltonian.
Section \ref{s:inf} contains the analysis of the behaviours of solutions near the infinity set in the space of initial values. 
A summary of the notation used in this paper is given in Appendix \ref{sec:notation}.

%% file: 2-space/space.tex
Since the third Painlev\'e equation is a second-order ordinary differential equation, solutions are locally defined by two initial values. Therefore, the space of initial values is two complex-dimensional. However, standard existence and uniqueness theorems only cover initial values of $y$ that are not arbitrarily small or large (where the second derivative given by $\PIII$ becomes ill-defined). In this section, we explain how to construct a regularized, compactified space of all possible initial values that includes such values and overcomes these issues.

We start by formulating $\PIII$ as an equivalent system of first-order equations in Section \ref{sec:system} and describing its autonomous limiting form obtained as $x\to\infty$ in Section \ref{sec:auto}. The mathematical construction of the space of initial values is then given in Section \ref{sec:okamoto}. Where $y$ is arbitrarily close to the singular values $0$ or $\infty$, the solutions have singular power series expansions, which become regularized Taylor expansions in corresponding domains of the initial value space. These regular expansions are provided in Section \ref{sec:movable}.

%% file: 2-space/system.tex
The third Painlev\'e equation is equivalent to the following system:
\begin{equation}\label{eq:p3-system}
\begin{aligned}
y'&=\frac{y^2z}{2}+\delta_1-\frac{y}{x},
\\
z'&=-\frac{yz^2}{2}+2\gamma y+\frac{2}{x}\left(\alpha+\frac{\beta_1}{y^2}\right),
\end{aligned}
\end{equation}
with $\delta_1=i\sqrt{\delta}$, $\beta_1=\beta-2\delta_1$, where $y$ is a solution of the third Painlev\'e equation ($\PIII$) and
$$
z=\frac2{y^2}\left(y'-\delta_1+\frac{y}{x}\right).
$$

For comparison, Okamoto represented the third Painlev\'e equation, in \cite{Okamoto1979}, as the following system:
\begin{equation}\label{eq:okamoto}
\begin{aligned}
y'&=\frac1x\left(y^2 z_o -(\theta_{\infty}xy^2+\eta_0y-\theta_0x)  \right),
\\
z_o'&=-\frac1x\left(y z_o^2-(2\theta_{\infty}xy+\eta_0)z_o+\theta_{\infty}(\eta_0+\eta_{\infty})x\right),
\end{aligned}
\end{equation}
with
$$
\gamma=\theta_{\infty}^2,
\quad
\delta=-\theta_0^2,
\quad
\alpha=-\theta_{\infty}(1+\eta_{\infty}),
\quad
\beta=\theta_0(1+\eta_0).
$$
Here, $y(x)$ is again a solution of ($\PIII$), while $z_o$ is given by
$$
z_o=\frac{1}{y^2}\left(x y'+\theta_{\infty}xy^2+\eta_0y-\theta_0x\right).
$$
The auxilliary function $z_o$ is related to $z$ in equations
\eqref{eq:p3-system} by the equation:
$$
z=-\frac{2 \left(\theta_{\infty} x^2 (\theta_{\infty} (\eta_0+\eta_{\infty})+z_{o}')-x z_o ((\eta_{\infty}+2) \theta_{\infty}+z_{o}')+z_o^2\right)}{x (x (\theta_{\infty} (\eta_0+\eta_{\infty})+z_{o}')-\eta_0 z_o)}.
$$

%% file: 2-space/auto.tex
In the limit $x\to\infty$, $\PIII$ becomes:
\begin{equation}\label{eq:p3-auto}
y''=\frac{y'^2}{y}+\gamma y^3+\frac{\delta}y,
\end{equation}
which is equivalent to the system:
\begin{equation}\label{eq:p3-system-auto}
\begin{aligned}
y'&=\frac{y^2z}{2}+\delta_1,
\\
z'&=-\frac{y z^2}{2}+2\gamma y,
\end{aligned}
\end{equation}
with $\delta_1=i\sqrt{\delta}$.

The system (\ref{eq:p3-system-auto}) is Hamiltonian, that is, it is
equivalent to
\[
y'=\frac{\partial E}{\partial z},\quad z'=-\frac{\partial E}{\partial y},
\]
with the Hamiltonian given by
\begin{equation}\label{eq:energy}
E(y,z)=\frac{y^2z^2}{4}+\delta_1 z-\gamma y^2,
\end{equation}
Equivalently, the autonomous equation (\ref{eq:p3-auto}) is equivalent to the following family of first order equations:
$$
y'^2=\gamma y^4+Cy^2-\delta,
\qquad
C=\const.
$$

Note that the flow (\ref{eq:p3-system-auto}) has four fixed points:
\begin{equation}\label{eq:auto_equi}
(y_e, z_e)=\begin{cases}
\left(i\sqrt{\frac{\delta_1}{\sqrt{\gamma}}} ,2\sqrt{\gamma}\right),\\
\left(-i\sqrt{\frac{\delta_1}{\sqrt{\gamma}}} ,2\sqrt{\gamma}\right),\\
\left(\sqrt{\frac{\delta_1}{\sqrt{\gamma}}} ,-2\sqrt{\gamma}\right),\\
\left(-\sqrt{\frac{\delta_1}{\sqrt{\gamma}}} ,-2\sqrt{\gamma}\right).
\end{cases}
\end{equation}
If $\delta_1=0$, then, in addition, all points of the line $y=0$ are fixed.

%% file: 2-space/okamoto.tex
In this section, we explain how to construct the space of initial values for the system (\ref{eq:p3-system}).
The notion of initial value spaces in Definition \ref{def:initial-values-space} is based on foliation theory, and we start by first motivating the reason for this construction.
We then explain how to construct such a space by carrying out resolutions or blow-ups, based on the process described in Definition \ref{def:blow-up}.

The system (\ref{eq:p3-system}) is a system of two first-order ordinary differential equations for $(y(x), z(x))$.
Given initial values $(y_0, z_0)$ at $x_0$, local existence and uniqueness theorems provide a solution that is defined on a local polydisk $U\times V$ in $\mathbb C\times \mathbb C^2$, where $x_0\in U\subset \mathbb C\setminus\{0\}$ and
$(y_0, z_0)\in V\subset (\mathbb C\setminus\{0\})\times\mathbb{C}$.
Our interest lies in global extensions of these local solutions.

However, the occurence of movable poles in the Painlev\'e transcendents acts as a barrier to the extension of $U\times V$ to the whole domain of (\ref{eq:p3-system}).
The first step to overcome this obstruction is to compactify the space $\mathbb C^2$, in order to include the poles.
We carry this out by embedding $\mathbb C^2$ into $\mathbb C\mathbb P^2$ and explicitly representing the system (\ref{eq:p3-system}) in the three affine coordinate charts of $\mathbb C\mathbb P^2$, which are given in Section \ref{sec:affine}.
The second step in this process results from the occurence of singularities in $V$ in the Painlev\'e vector field (\ref{eq:p3-system}).

By the term {\em singularity} we mean points where $(dy/dx, dz/dx)$ becomes either unbounded or undefined because at least one component approaches the undefined limit $0/0$. We are led therefore to construct a space in which the points where the singularities are regularised. The process of regularisation is called "blowing up" or \emph{resolving a singularity}.

In $\mathbb C\mathbb P^2$, these singularities occur in the $(y_{02},z_{02})$ and $(y_{03},z_{03})$ charts.  The appearence of these singularities is related to the irreducibility of the solutions of Painlev\'e equations as indicated by the following theorem, due to Painlev\'e. 
\begin{theorem}[\cite{Painleve1897}]
If the space of initial values for a differential equation is a compact rational variety, then the equation can be reduced either to a linear differential equation of higher order or to an equation for elliptic functions.
\end{theorem}
It is well known that the solutions of Painlev\'e equations are irreducible (in the sense of the theorem). Since $\mathbb{CP}^2$ is a compact rational variety, the theorem implies $\mathbb{CP}^2$ cannot be the space of initial values for (\ref{eq:p3-system}).

We are now in a position to define the notion of initial value space.

\begin{definition}[\cite{Gerard1975}, \cite{GerardSec1972,Gerard1983,Okamoto1979}]\label{def:initial-values-space} 
Let $(\mathcal{E},\pi,\mathcal{B})$ be a complex analytic fibration, $\Phi$ its foliation, and $\Delta$ a holomorphic differential system on $\mathcal{E}$, such that:
\begin{itemize}
\item the leaves of $\Phi$ correspond to the solutions of $\Delta$;
\item the leaves of $\Phi$ are transversal to the fibres of $\mathcal{E}$;
\item for each path $p$ in the base $\mathcal{B}$ and each point $X\in \mathcal{E}$, such that $\pi(X)\in p$, the path $p$ can be lifted into the leaf of $\Phi$ containing point $X$.
\end{itemize}
Then each fibre of the fibration is called \emph{a space of initial values} for the system $\Delta$.
\end{definition}

The properties listed in Definition \ref{def:initial-values-space} imply that each leaf of the foliation is isomorphic to the base $\mathcal{B}$.
Since the transcendental solutions of the third Painlev\'e equation can be globally extended as meromorphic functions on $\mathbb{C}\setminus\{0\}$ \cite{JK1994,HL2001-3}, we search for the fibration with the base equal to $\mathbb{C}\setminus\{0\}$. 

In order to construct the fibration, we apply the blow-up procedure, defined below \cite{HartshorneAG,GrifHarPRINC,DuistermaatBOOK} to the singularities of the system  (\ref{eq:p3-system}) that occur where at least one component becomes undefined of the form $0/0$.
Okamoto \ocite{Okamoto1979} showed that such singular points are contained in the closure of infinitely many leaves.
Moreover, these leaves are holomorphically extended at such a point.

\begin{definition}\label{def:blow-up}
\emph{The blow-up} of the plane $\mathbb{C}^2$ at point $(0,0)$ is the closed subset $X$ of $\mathbb{C}^2\times\mathbb{CP}^1$ defined by the equation $u_1t_2=u_2t_1$, where $(u_1,u_2)\in\mathbb{C}^2$ and $[t_1:t_2]\in\mathbb{CP}^1$, see Figure \ref{fig:blow-up}.
There is a natural morphism $\varphi: X\to\mathbb{C}^2$, which is the restriction of the projection from $\mathbb{C}^2\times\mathbb{CP}^1$ to the first factor.
$\varphi^{-1}(0,0)$ is the projective line $\{(0,0)\}\times\mathbb{CP}^1$, called \emph{the exceptional line}.
\end{definition}
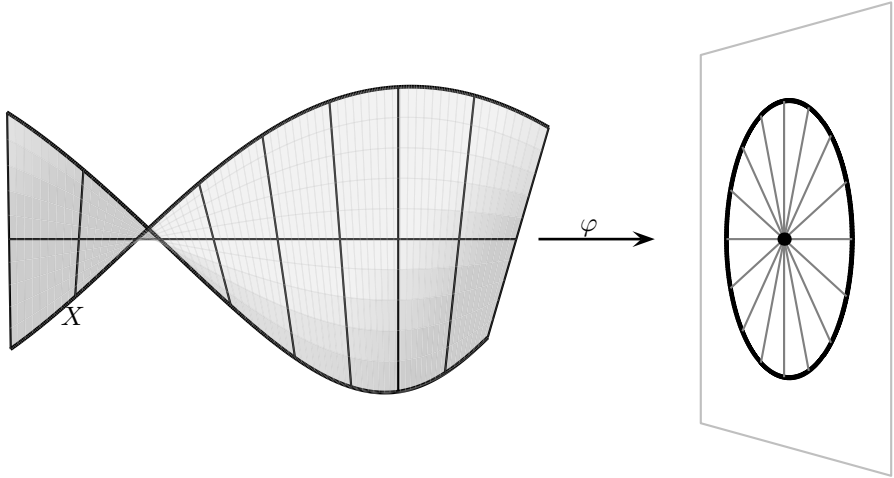
\begin{figure}[h]
\input{0-figures/blow-up}
\caption{The blow-up of the plane at a point.}\label{fig:blow-up}
\end{figure}

\begin{remark}
Notice that the points of the exceptional line $\varphi^{-1}(0,0)$ are in bijective correspondence with the lines containing $(0,0)$.
On the other hand,
$\varphi$ is an isomorphism between $X\setminus\varphi^{-1}(0,0)$ and $\mathbb{C}^2\setminus\{(0,0)\}$.
More generally, any complex two-dimensional surface can be blown up at a point \cite{HartshorneAG,GrifHarPRINC,DuistermaatBOOK}.
In a local chart around that point, the construction will look the same as described for the case of the plane.
\end{remark}

Notice that the blow-up construction separates the lines containing the point $(0,0)$ in Definition \ref{def:blow-up}, as shown in Figure \ref{fig:blow-up}.
In this way, the solutions of \eqref{eq:p3-system} containing the same point can be separated.
Additional blow-ups may be required if the solutions have a commont tangent line or a tangency of higher order at such a point.
The explicit resolution of the vector field \eqref{eq:p3-system} is carried out in Appendix \ref{sec:p3resolution}.
We show that the process requires 11 resolutions of singularities, or, blow-ups.

The resulting surface $\mathcal{D}(x)$ is shown in Figure \ref{fig:okamoto}.
\begin{figure}[h]
\centering
\input{0-figures/fiber-p3}
\caption{The fiber $\mathcal{D}(x)$ of the  Okamoto space.}\label{fig:okamoto}
\end{figure}
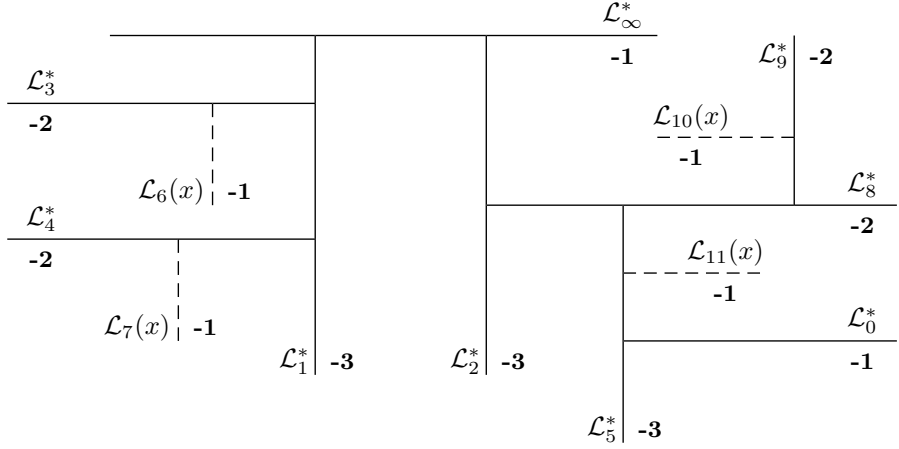
In Figure \ref{fig:okamoto}, $\mathcal{L}_{\infty}^*$ and $\mathcal{L}_0^*$ are the proper preimages of the line at the infinity and the line $y=0$, while $\mathcal{L}_{1}^*$, $\mathcal{L}_{2}^*$, $\mathcal{L}_{3}^*$, $\mathcal{L}_{4}^*$, $\mathcal{L}_{5}^*$, $\mathcal{L}_{8}^*$, $\mathcal{L}_{9}^*$  are the proper preimages of the exceptional lines obtained by blow ups at points $b_1$, $b_2$, $b_3$, $b_4$, $b_5$, $b_8$, $b_9$ respectively,
and $\mathcal{L}_{6}(x)$, $\mathcal{L}_{7}(x)$, $\mathcal{L}_{10}(x)$, $\mathcal{L}_{11}(x)$ are exceptional lines obtained by blowing up points $b_5$, $b_6$, $b_{10}$, $b_{11}$ respectively. 
The self-intersection number of each line, after all blow-ups are completed, is indicated in the figure.

Okamoto described so called \emph{singular points of the first class} that are not contained in the closure of any leaf of the foliation given by the system of differential equations.
At such points, the corresponding vector field is infinite.
For example, in the first affine chart $(y,z)$, such singular points are given by $y=0$.
In the surface $\mathcal{D}(x)$, all remaining singular points are of the first class, and the fibre of the initial value space is obtained by removing them:
$$
\mathcal{E}(x)=\mathcal{D}(x)\setminus\left(\bigcup_{j=0}^5\mathcal{L}_j^*\cup\mathcal{L}_8^*\cup\mathcal{L}_9^*\cup\mathcal{L}_{\infty}\right).
$$

In $\mathcal{D}(x)$, each line with self-intersection number $-1$ can be blown down again. 
Blowing down $\mathcal{L}_{\infty}^*$ and $\mathcal{L}_{0}^*$, we get the surface $\mathcal{F}(x)$, which is shown in Figure \ref{fig:okamoto-blow-down}.
\begin{figure}[h]
\centering
\input{0-figures/fiber-p3-blowdown}
\caption{The fiber $\mathcal{F}(x)$ of the Okamoto space, after two blow downs.}\label{fig:okamoto-blow-down}
\end{figure}
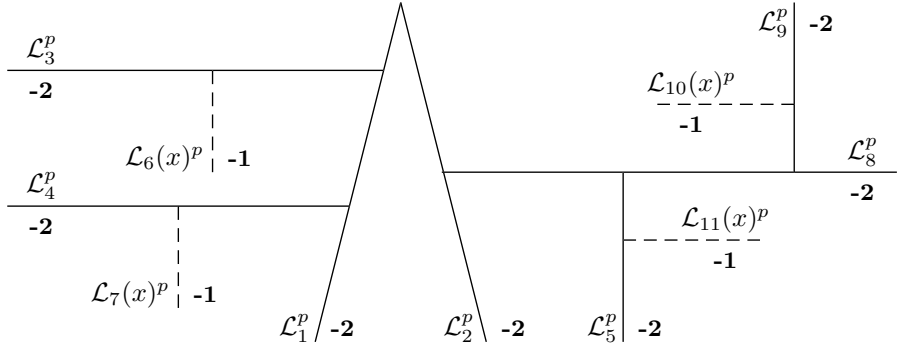
The projection of each remaining line from $\mathcal{D}(x)$ is denoted by the same index but now with superscript $p$.
Notice that the self-intersection numbers of $\mathcal{L}_{1}^p$, $\mathcal{L}_{2}^p$, $\mathcal{L}_{5}^p$ are no longer the same as of the corresponding pre-images $\mathcal{L}_{1}^*$, $\mathcal{L}_{2}^*$, $\mathcal{L}_{5}^*$.
In this space,  we denote by $\mathcal{I}$ the set of all singular points of the first class in $\mathcal{F}(x)$:
$$
\mathcal{I}=\bigcup_{j=1}^5\mathcal{L}_j^p\cup\mathcal{L}_8^p\cup\mathcal{L}_9^p.
$$
The fibre $\mathcal{E}(x)$ of the initial value space can be identified with $\mathcal{F}(x)\setminus\mathcal{I}$.

If $\mathcal{S}$ is a surface obtained from the projective plane by a several successive blow-ups of points, then the group of all automorphisms of the Picard group $\Pic(\mathcal{S})$ preserving the canonical divisor $K$ is generated by the reflections
$X\mapsto X+(X.\omega)\omega$,
with $\omega\in\Pic(\mathcal{S})$, $K.\omega=0$, $\omega.\omega=-2$.
That group is an affine Weyl group and the lines of self-intersection $-2$ are its simple roots.
Representing each such line by a node and connecting a pair of nodes by a line only if they intersect in the fibre, we obtain the Dynkin diagram shown in Figure \ref{fig:dynkin}, which is of type $D_6^{(1)}$.
For detailed expositions on the topic of surfaces and root systems, see \cite{Demazure1980, Harbourne1985} and references therein.
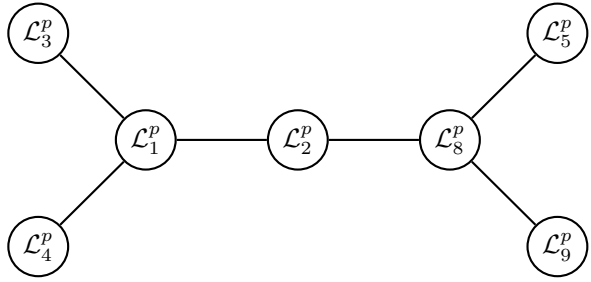
\begin{figure}[h]
\centering
\input{0-figures/dynkin-p3}
\caption{The Dynkin diagram of $D_6^{(1)}$.}\label{fig:dynkin}
\end{figure}

In the limit $x\to\infty$, the resulting Okamoto space is compactified by the fibre $\mathcal{F}(\infty)$, corresponding to the autonomous system (\ref{eq:p3-system-auto}), see Figure \ref{fig:okamoto-auto}.
\begin{figure}[h]
\centering
\input{0-figures/fiber-p3-limit}
\caption{The fiber $\mathcal{F}(\infty)$ of the  Okamoto space.}\label{fig:okamoto-auto}
\end{figure}
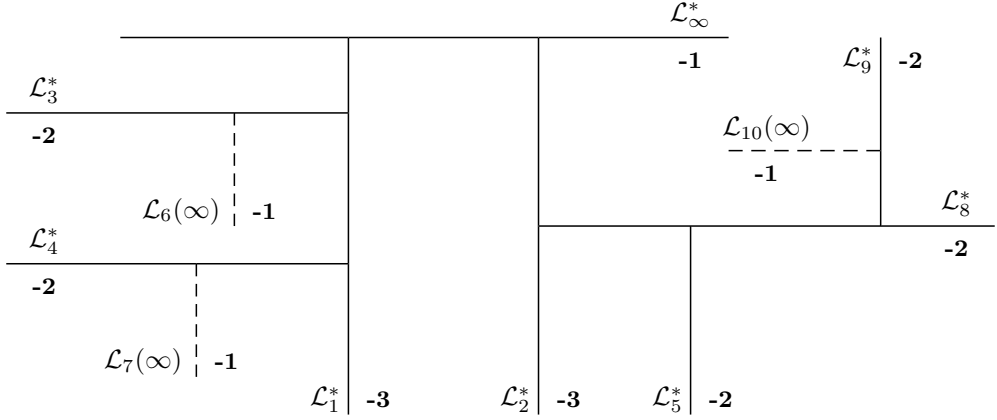
We observe that the resolution process will be different for the autonomous system (\ref{eq:p3-system-auto}).
The reason is that the affine part of the line $u=0$ consists only of regular points of the system (\ref{eq:p3-system-auto}) and the point $b_{11}$ does not appear as a base point there.
Thus, the resolution of singularities of  (\ref{eq:p3-system-auto}) is achieved after ten blow-ups.
Blowing down $\mathcal{L}_{\infty}^*$ in $\mathcal{F}(\infty)$, we get another representation of the fiber $\mathcal{F}(\infty)$, which is shown in Figure \ref{fig:okamoto-blow-down-auto}.
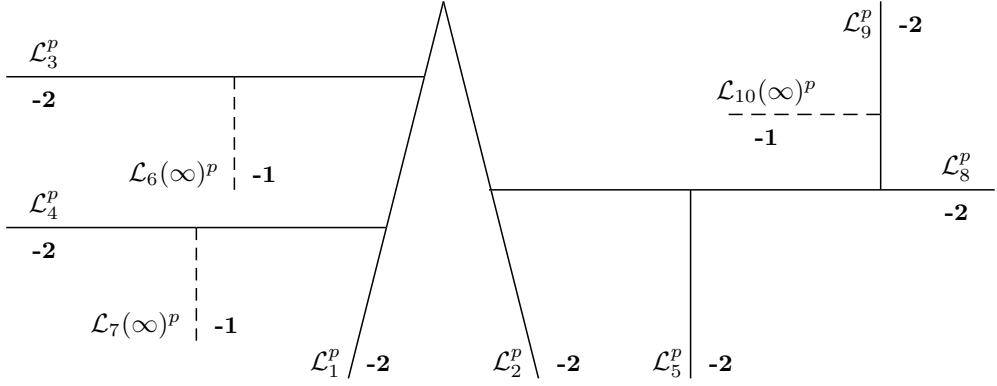
\begin{figure}[h]
\centering
\input{0-figures/fiber-p3-blowdown-auto}
\caption{The fiber $\mathcal{F}(\infty)$ of the Okamoto space, after one blow down.}\label{fig:okamoto-blow-down-auto}
\end{figure}

%% file: 0-figures/blow-up.tex
\begin{pspicture*}(-9.2,-4)(4.7,4)

\psset{unit=0.51}

\psset{viewpoint=4 10 0,Decran=40,lightsrc=20 20 20}

 \psline[linecolor=black,fillcolor=black,incolor=black,linewidth=1pt,arrows=->,arrowsize=2pt 4](-3,0)(0,0)
\rput(-1.7,0.3){$\varphi$}

\rput(-15,-2){$X$}

\psSolid[object=plan, linecolor=gray!50, definition=equation, args={[1 0 0 1]}, base=-1.5 1.5 -1.5 1.5]

\defFunction[algebraic]{krug}(t)
{0}{cos(t)}{sin(t)}
\psSolid[object=courbe,linecolor=black,r=0.01,range=0 pi 2 mul,
linewidth=0.06,resolution=1000,
function=krug](-1,0,0)

\psSolid[object=line, linecolor=gray, args=-1 1 0 -1 -1 0]
\psSolid[object=line, linecolor=gray, args=-1 0 1 -1  0 -1]
\psSolid[object=line, linecolor=gray, args=-1 0.707 0.707 -1  -0.707 -0.707]
\psSolid[object=line, linecolor=gray, args=-1 -0.707 0.707 -1  0.707 -0.707]
\psSolid[object=line, linecolor=gray, args=-1 -0.383 0.924 -1  0.383 -0.924]
\psSolid[object=line, linecolor=gray, args=-1 0.383 0.924 -1  -0.383 -0.924]
\psSolid[object=line, linecolor=gray, args=-1  0.924 -0.383 -1   -0.924 0.383]
\psSolid[object=line, linecolor=gray, args=-1  -0.924 -0.383 -1   0.924 0.383]

\psSolid[object=point, args=-1 0 0]

\defFunction[algebraic]{helix1}(h)
{h+1}{cos(h+pi/4)}{sin(h+pi/4)}
\psSolid[object=courbe,r=0,
range=0 pi, linecolor=black, linewidth=0.1, resolution=360, function=helix1]

\defFunction[algebraic]{helix2}(h)
{h+1}{-cos(h+pi/4)}{-sin(h+pi/4)}
\psSolid[object=courbe,r=0,
range=0 pi, linecolor=black, linewidth=0.1, resolution=360, function=helix2]%

\psSolid[object=line, args=1 0 0 1 pi add 0 0]

\psSolid[object=line, linecolor=black, args=1 0.707 0.707 1 -0.707 -0.707 ]
\psSolid[object=line, linecolor=black, args=2.571 -0.707 0.707 2.571 0.707 -0.707 ]
\psSolid[object=line, linecolor=black, args=2.178 -0.383 0.924 2.178  0.383 -0.924 ]
\psSolid[object=line, linecolor=black, args=1.393 -0.383 -0.924 1.393  0.383 0.924 ]
\psSolid[object=line, linecolor=black, args=2.963  0.924 -0.383 2.963   -0.924 0.383]
\psSolid[object=line, linecolor=black, args=3.75  -0.924 -0.383 3.75  0.924 0.383]
\psSolid[object=line, linecolor=black, args=1.8   0 1 1.8  0 -1 ]
\psSolid[object=line, linecolor=black, args=4.14 0.707 0.707 4.14 -0.707 -0.707 ]

\defFunction[algebraic]{helix}(t,h)
{h+1}{t*cos(h+pi/4)}{t*sin(h+pi/4)}
\psSolid[object=surfaceparametree,linewidth=1sp,linecolor=gray!50,
     base=-1 1 0 pi,fillcolor=gray!50,incolor=gray!50,opacity=0.2,
     function=helix,
   ngrid=10 72]

\end{pspicture*}

%% file: 0-figures/fiber-p3.tex
\begin{pspicture}(-1.4,-6)(11.9,0.5)

\psset{unit=0.9}

\psset{linecolor=black,linewidth=0.02,fillstyle=solid}

\psline(0,0)(8,0)
\rput(7.5,0.3){$\mathcal{L}_{\infty}^*$}
\rput(7.5,-0.3){\small\textbf{-1}}

\psline(3,0)(3,-5)
\rput(2.7,-4.8){$\mathcal{L}_{1}^*$}
\rput(3.4,-4.8){$\small\textbf{-3}$}

\psline(3,-1)(-1.5,-1)
\rput(-1,-0.7){$\mathcal{L}_{3}^*$}
\rput(-1,-1.3){\small\textbf{-2}}

\psline[linestyle=dashed](1.5,-1)(1.5,-2.5)
\rput(0.9,-2.3){$\mathcal{L}_{6}(x)$}
\rput(1.9,-2.3){\small\textbf{-1}}

\psline(3,-3)(-1.5,-3)
\rput(-1,-2.7){$\mathcal{L}_{4}^*$}
\rput(-1,-3.3){\small\textbf{-2}}

\psline[linestyle=dashed](1,-3)(1,-4.5)
\rput(0.4,-4.3){$\mathcal{L}_{7}(x)$}
\rput(1.4,-4.3){\small\textbf{-1}}

\psline(5.5,0)(5.5,-5)
\rput(5.2,-4.8){$\mathcal{L}_{2}^*$}
\rput(5.9,-4.8){\small\textbf{-3}}

\psline(5.5,-2.5)(11.5,-2.5)
\rput(11,-2.2){$\mathcal{L}_{8}^*$}
\rput(11,-2.8){\small\textbf{-2}}

\psline(10,-2.5)(10,0)
\rput(9.7,-0.3){$\mathcal{L}_{9}^{*}$}
\rput(10.4,-0.3){\small\textbf{-2}}

\psline[linestyle=dashed](10,-1.5)(8,-1.5)
\rput(8.5,-1.2){$\mathcal{L}_{10}(x)$}
\rput(8.5,-1.8){\small\textbf{-1}}

\psline(7.5,-2.5)(7.5,-6)
\rput(7.2,-5.8){$\mathcal{L}_{5}^{*}$}
\rput(7.9,-5.8){\small\textbf{-3}}

\psline(7.5,-4.5)(11.5,-4.5)
\rput(11,-4.2){$\mathcal{L}_{0}^{*}$}
\rput(11,-4.8){\small\textbf{-1}}

\psline[linestyle=dashed](7.5,-3.5)(9.5,-3.5)
\rput(9,-3.2){$\mathcal{L}_{11}(x)$}
\rput(9,-3.8){\small\textbf{-1}}

\end{pspicture}

%% file: 0-figures/fiber-p3-blowdown.tex
\begin{pspicture}(-1.4,-6)(11.9,0.5)

\psset{unit=0.9}

\psset{linecolor=black,linewidth=0.02,fillstyle=solid}

\psline(4.25,0)(3,-5)
\rput(2.7,-4.8){$\mathcal{L}_{1}^p$}
\rput(3.4,-4.8){$\small\textbf{-2}$}

\psline(4,-1)(-1.5,-1)
\rput(-1,-0.7){$\mathcal{L}_{3}^p$}
\rput(-1,-1.3){\small\textbf{-2}}

\psline[linestyle=dashed](1.5,-1)(1.5,-2.5)
\rput(0.8,-2.3){$\mathcal{L}_{6}(x)^p$}
\rput(1.9,-2.3){\small\textbf{-1}}

\psline(3.5,-3)(-1.5,-3)
\rput(-1,-2.7){$\mathcal{L}_{4}^p$}
\rput(-1,-3.3){\small\textbf{-2}}

\psline[linestyle=dashed](1,-3)(1,-4.5)
\rput(0.3,-4.3){$\mathcal{L}_{7}(x)^p$}
\rput(1.4,-4.3){\small\textbf{-1}}

\psline(4.25,0)(5.5,-5)
\rput(5.1,-4.8){$\mathcal{L}_{2}^p$}
\rput(5.9,-4.8){\small\textbf{-2}}

\psline(4.85,-2.5)(11.5,-2.5)
\rput(11,-2.2){$\mathcal{L}_{8}^p$}
\rput(11,-2.8){\small\textbf{-2}}

\psline(10,-2.5)(10,0)
\rput(9.7,-0.3){$\mathcal{L}_{9}^p$}
\rput(10.4,-0.3){\small\textbf{-2}}

\psline[linestyle=dashed](10,-1.5)(8,-1.5)
\rput(8.5,-1.2){$\mathcal{L}_{10}(x)^p$}
\rput(8.5,-1.8){\small\textbf{-1}}

\psline(7.5,-2.5)(7.5,-5)
\rput(7.2,-4.8){$\mathcal{L}_{5}^p$}
\rput(7.9,-4.8){\small\textbf{-2}}

\psline[linestyle=dashed](7.5,-3.5)(9.5,-3.5)
\rput(9,-3.2){$\mathcal{L}_{11}(x)^p$}
\rput(9,-3.8){\small\textbf{-1}}

\end{pspicture}

%% file: 0-figures/dynkin-p3.tex
\begin{pspicture}(-2,-2)(6,2)

\psset{linecolor=black}
\rput(0,0){\circlenode{1}{$\mathcal{L}_1^p$}}
\rput(2,0){\circlenode{2}{$\mathcal{L}_2^p$}}
\rput(4,0){\circlenode{8}{$\mathcal{L}_8^p$}}

\rput(-1.414,1.414){\circlenode{3}{$\mathcal{L}_3^p$}}
\rput(-1.414,-1.414){\circlenode{4}{$\mathcal{L}_4^p$}}

\rput(5.414,1.414){\circlenode{5}{$\mathcal{L}_5^p$}}
\rput(5.414,-1.414){\circlenode{9}{$\mathcal{L}_9^p$}}

\ncline{-}{1}{2}
\ncline{-}{1}{3}
\ncline{-}{1}{4}
\ncline{-}{2}{8}
\ncline{-}{8}{9}
\ncline{-}{8}{5}
    
\end{pspicture}

%% file: 0-figures/fiber-p3-limit.tex
\begin{pspicture}(-1.5,-5)(12,0.5)

\psset{linecolor=black,linewidth=0.02,fillstyle=solid}

\psline(0,0)(8,0)
\rput(7.5,0.3){$\mathcal{L}_{\infty}^*$}
\rput(7.5,-0.3){\small\textbf{-1}}

\psline(3,0)(3,-5)
\rput(2.7,-4.8){$\mathcal{L}_{1}^*$}
\rput(3.4,-4.8){$\small\textbf{-3}$}

\psline(3,-1)(-1.5,-1)
\rput(-1,-0.7){$\mathcal{L}_{3}^*$}
\rput(-1,-1.3){\small\textbf{-2}}

\psline[linestyle=dashed](1.5,-1)(1.5,-2.5)
\rput(0.8,-2.3){$\mathcal{L}_{6}(\infty)$}
\rput(1.9,-2.3){\small\textbf{-1}}

\psline(3,-3)(-1.5,-3)
\rput(-1,-2.7){$\mathcal{L}_{4}^*$}
\rput(-1,-3.3){\small\textbf{-2}}

\psline[linestyle=dashed](1,-3)(1,-4.5)
\rput(0.3,-4.3){$\mathcal{L}_{7}(\infty)$}
\rput(1.4,-4.3){\small\textbf{-1}}

\psline(5.5,0)(5.5,-5)
\rput(5.2,-4.8){$\mathcal{L}_{2}^*$}
\rput(5.9,-4.8){\small\textbf{-3}}

\psline(5.5,-2.5)(11.5,-2.5)
\rput(11,-2.2){$\mathcal{L}_{8}^*$}
\rput(11,-2.8){\small\textbf{-2}}

\psline(10,-2.5)(10,0)
\rput(9.7,-0.3){$\mathcal{L}_{9}^{*}$}
\rput(10.4,-0.3){\small\textbf{-2}}

\psline[linestyle=dashed](10,-1.5)(8,-1.5)
\rput(8.5,-1.2){$\mathcal{L}_{10}(\infty)$}
\rput(8.5,-1.8){\small\textbf{-1}}

\psline(7.5,-2.5)(7.5,-5)
\rput(7.2,-4.8){$\mathcal{L}_{5}^{*}$}
\rput(7.9,-4.8){\small\textbf{-2}}

\end{pspicture}

%% file: 0-figures/fiber-p3-blowdown-auto.tex
\begin{pspicture}(-1.5,-5)(12,0.5)

\psset{linecolor=black,linewidth=0.02,fillstyle=solid}

\psline(4.25,0)(3,-5)
\rput(2.7,-4.8){$\mathcal{L}_{1}^p$}
\rput(3.4,-4.8){$\small\textbf{-2}$}

\psline(4,-1)(-1.5,-1)
\rput(-1,-0.7){$\mathcal{L}_{3}^p$}
\rput(-1,-1.3){\small\textbf{-2}}

\psline[linestyle=dashed](1.5,-1)(1.5,-2.5)
\rput(0.7,-2.3){$\mathcal{L}_{6}(\infty)^p$}
\rput(1.9,-2.3){\small\textbf{-1}}

\psline(3.5,-3)(-1.5,-3)
\rput(-1,-2.7){$\mathcal{L}_{4}^p$}
\rput(-1,-3.3){\small\textbf{-2}}

\psline[linestyle=dashed](1,-3)(1,-4.5)
\rput(0.2,-4.3){$\mathcal{L}_{7}(\infty)^p$}
\rput(1.4,-4.3){\small\textbf{-1}}

\psline(4.25,0)(5.5,-5)
\rput(5.1,-4.8){$\mathcal{L}_{2}^p$}
\rput(5.9,-4.8){\small\textbf{-2}}

\psline(4.85,-2.5)(11.5,-2.5)
\rput(11,-2.2){$\mathcal{L}_{8}^p$}
\rput(11,-2.8){\small\textbf{-2}}

\psline(10,-2.5)(10,0)
\rput(9.7,-0.3){$\mathcal{L}_{9}^p$}
\rput(10.4,-0.3){\small\textbf{-2}}

\psline[linestyle=dashed](10,-1.5)(8,-1.5)
\rput(8.5,-1.2){$\mathcal{L}_{10}(\infty)^p$}
\rput(8.5,-1.8){\small\textbf{-1}}

\psline(7.5,-2.5)(7.5,-5)
\rput(7.2,-4.8){$\mathcal{L}_{5}^p$}
\rput(7.9,-4.8){\small\textbf{-2}}

\end{pspicture}

%% file: 2-space/movable.tex
Here, we consider neighbourhoods of exceptional lines where the Painlev\'e vector field (\ref{eq:p3-system}) becomes unbounded.
The construction given in Appendix \ref{sec:p3resolution} shows that these are given by the lines $\mathcal{L}_6$, $\mathcal{L}_7$, $\mathcal{L}_{10}$ and $\mathcal{L}_{11}$.

\subsubsection*{Movable poles of $y$} 

The line $\mathcal{L}_6$ is obtained after three consecutive blow ups,
at $b_1$, $b_3$, and $b_6$.
The flow of (\ref{eq:p3-system}) passes through $b_1$ at a given point $x_0$ if $y$ has a pole at $x_0$ and $z$ is either regular at $x_0$ or has a pole of smaller order than $y$.
Next, the flow passes through $b_3$ if $z(x_0)=\sqrt{4\gamma}$.
From the first equation of (\ref{eq:p3-system}), we get then that $y$
has a simple pole at $x_0$.

Corresponding Laurent expansions of $y$ and $z$ around $x_0$ can be calculated from Equations (\ref{eq:p3-system}):
$$
\begin{aligned}
y&=-\frac{1}{\sqrt{\gamma}(x-x_0)}-\frac{\alpha-\sqrt{\gamma}}{2x_0\gamma}+B(x-x_0)+\dots,
\\
z&=\sqrt{4\gamma}-\frac{2\alpha}{x_0}(x-x_0)
+\left( 2 \gamma(3B-  \delta_1) +\frac{3 \alpha ^2-4 \alpha  \sqrt{\gamma }+5 \gamma }{2 \sqrt{\gamma } x_0^2} \right)(x-x_0)^2+\dots,
\end{aligned}
$$
where $B$ is an arbitrary constant.
Each value of $B$ corresponds to exactly one point of $\mathcal{L}_6\setminus\mathcal{I}$.

The alternative sign of the square root of $\gamma$ leads to another set of simple movable poles of $y$, with $z(x_0)=-\sqrt{4\gamma}$:$$
\begin{aligned}
y&=\frac{1}{\sqrt{\gamma}(x-x_0)}-\frac{\alpha+\sqrt{\gamma}}{2x_0\gamma}+B(x-x_0)+\dots,
\\
z&=-\sqrt{4\gamma}-\frac{2\alpha}{x_0}(x-x_0)
+\left( 2 \gamma(3B-  \delta_1) -\frac{3 \alpha ^2+4 \alpha  \sqrt{\gamma }+5 \gamma }{2 \sqrt{\gamma } x_0^2} \right)(x-x_0)^2+\dots.
\end{aligned}
$$
Values of the arbitrary constant $B$ now correspond to points of the line $\mathcal{L}_7\setminus\mathcal{I}$.

\subsubsection*{Movable poles of $z$} 

The line $\mathcal{L}_{10}$ is obtained after five consecutive blow ups, at $b_2$, $b_5$, $b_8$, $b_9$, and $b_{10}$.
The flow of (\ref{eq:p3-system}) passes through $b_2$ at a given point $x_0$ if $z$ has a pole at $x_0$ and $y$ is either regular at $x_0$ or has a pole of smaller order than $z$.
Next, the flow passes through $b_5$ if $y(x_0)=0$.
Denote by $N$ the order of the zero of $y$ at $x_0$ and by $P$ the order of the pole of $z$ at $x_0$.
 
The flow passes through $b_8$ if $yz$ has a pole at $x_0$, that is, $P> N$.
The flow passes through $b_9$ if $y^2z$ is equal to $-4\delta_1$ at $x_0$, thus $2N=P$.
On the other hand, from the second equation of (\ref{eq:p3-system}), we get $P+1=2P-N$, that is $P=N+1$, which means $N=1$ and $P=2$.

So, we concluded that $z$ has a double pole and $y$ a simple zero at $x_0$.
The coefficients in the Laurent expansions of $y$ and $z$ around $x_0$ can be calculated from  (\ref{eq:p3-system}):
$$
\begin{aligned}
y&=-\delta_1 (x-x_0)-\frac{\beta_1 +3 \delta_1 }{2 x_0}(x-x_0)^2+B(x-x_0)^3+\dots,
\\
z&=-\frac{4}{\delta_1(x-x_0)^2}+ \frac{2 (\beta_1 +2 \delta_1 )}{\delta_1^2 x_0}\cdot \frac{1}{x-x_0}
-\left(
\frac{2 B}{\delta_1 ^2}+
\frac{\beta_1 ^2+5 \beta_1  \delta_1 +4 \delta_1 ^2}{\delta_1 ^3 x_0^2}
\right)
+\dots.
\end{aligned}
$$
Values of the arbitrary constant $B$ correspond to the points of the line $\mathcal{L}_{10}\setminus\mathcal{I}$.

The line $\mathcal{L}_{11}$ is obtained after three consecutive blow ups, at $b_2$, $b_5$, and $b_{11}$.
We have already concluded that the flow passes through $b_5$ if $z$
has a pole and $y$ has a zero.
The flow passes through $b_{11}$ if $yz$ has neither pole nor zero at the given point $x_0$, thus the order of pole of $z$ equals to the order of zero of $y$.
From the first equation of (\ref{eq:p3-system}), we conclude that $z$ has a simple pole and $y$ a simple zero.
Their Laurent expansions are:
$$
\begin{aligned}
y&=\delta_1 (x-x_0)-\frac{\beta_1+\delta_1}{2x_0}(x-x_0)^2+B(x-x_0)^3+\dots,
\\
z&=-\frac{2\beta_1}{\delta_1^2 x_0}\cdot\frac{1}{x-x_0}-\frac{2 \beta_1 ^2+3 \beta_1  \delta_1 +3 \delta_1^2-6 \delta_1  x_0^2 B}
{\delta_1 ^3 x_0^2}
+\dots.
\end{aligned}
$$
Values of the arbitrary constant $B$ now correspond to the points of the line $\mathcal{L}_{11}\setminus\mathcal{I}$.

%% file: 3-special/special.tex
In this section, we recall some facts about rational solutions of
$\PIII$ and relate them to equilibrium points of the total energy $E$,
defined by Equation \eqref{eq:energy}. Notice that although $E$ was
defined for the autonomous system, we can extend it to solutions
$\bigl(y(x), z(x)\bigr)$ of the $\PIII$
system \eqref{eq:p3-system}. Doing so gives
\begin{align}
\nonumber E'&=\frac{d E}{d x}=\frac{\partial E}{\partial y}y'+\frac{\partial E}{\partial z}z',
\\
\nonumber &=\frac1x
\left( 
 2 \alpha  \delta_1+\frac{2 \beta_1 \delta_1}{y^2}+\alpha  y^2 z+\beta_1 z+2 \delta_1 z -2E
\right),
\\
\label{eq:energychange}&=
\frac1x\left(
\left(2\alpha+\frac{\beta_1}{y^2}\right)\left(\frac{y^2z}{2}+\delta_1\right)
+
2\delta_1 z
-
2E
\right)
.
\end{align}

The pencil of elliptic curves arising from $E$
is given by:
\begin{equation}\label{eq:pencil}
h(y,z)\equiv h_c(y, z):=c+E=0,
\end{equation}
where $c$ is an arbitrary constant parameter. For general values of
$c$, the curves are non-singular and the corresponding curves will be smooth.
To investigate possible singularities of curves, consider the conditions:
$$
\frac{\partial{h}}{\partial y}=0,
\quad
\frac{\partial{h}}{\partial z}=0,
$$
which give:
$$
\frac{y z^2}{2}-2\gamma y=0,
\quad
\frac{y^2 z}{2}+\delta_1=0.
$$
Notice that these lead to fixed points of the autonomous flow
(\ref{eq:p3-system-auto}), given by Equation \eqref{eq:auto_equi}.

The first two points in Equation \eqref{eq:auto_equi} belong to the
curve corresponding to $c_1=-2\delta_1\sqrt{\gamma}$, while the last
two lie on curves given by $c_2=2\delta_1\sqrt{\gamma}$:
$$
\begin{aligned}
&c_1=-2\delta_1\sqrt{\gamma},
\quad
h_{c_1}(y,z)=\frac{1}{4}(z-2\sqrt{\gamma})(y^2 z+2y^2\sqrt{\gamma}+4\delta_1),
\\
&c_2=2\delta_1\sqrt{\gamma},
\quad
h_{c_2}(y,z)=\frac{1}{4}(z+2\sqrt{\gamma})(y^2 z-2y^2\sqrt{\gamma}+4\delta_1).
\end{aligned}
$$
It turns out that these fixed points correspond exactly to special
solutions of $\PIII$.

%% file: 3-special/rational.tex
It is well known that $\PIII$ has special rational solutions corresponding to
special values of parameters \cite[\S 32.8(iii)]{NISThandbook}. All
such can be generated by applying
B\"acklund transformations to a group of seed solutions given by
\begin{align}
&y=k,& &\text{for }& &\beta=-\alpha k^2,\ \delta=-\gamma k^4,\label{eq:constsol}
\\
&y=kx,& &\text{for }& &\alpha=0,\ \gamma=0,\ \delta=-\beta k,
\\
&y=\frac{x+k}{x+k+1},& &\text{for }& &\alpha=2k+1,\ \beta=-2k+1,\ \gamma=1,\delta=-1.
\end{align}
Notice that the constant solutions \eqref{eq:constsol} correspond to the fixed points of the autonomous system.

%% file: 4-infinity/infinity.tex
In this section, we study the behaviour of the solutions of the system (\ref{eq:p3-system}) near the set $\mathcal{I}$, where the vector field is infinite.
We prove that $\mathcal{I}$ is a repeller for the solutions and that each solution which comes sufficiently close to $\mathcal{I}$ at a certain point $x$ will have a pole in a neighbourhood of $x$.
The results of this section are summarised at the end in Corollary \ref{cor:infinity}, while estimates for behaviour near $\mathcal{I}$ are given in Theorem \ref{th:estimates}.

\begin{lemma}\label{lemma:L2}
For every $\epsilon>0$, there exists a neighbourhood $U$ of $\mathcal{L}_2^{p}$ such that
$$
\left|\frac{E'}{E}+\frac{2}{x}\right|<\epsilon
\quad
\text{in }
U.
$$
\end{lemma}

\begin{proof}
In the charts $(y_{21},z_{21})$ and $(y_{22},z_{22})$, we have:
$$
\frac{E'}{E}+\frac{2}{x}=
\begin{cases}
\dfrac{4 y_{21} z_{21} \left(\alpha +2 \alpha  \delta_1 y_{21}^3 z_{21}+y_{21}^2 \left(\beta_1+2 \delta_1+2 \beta_1 \delta_1 y_{21}^3 z_{21}\right)\right)}{x \left(4 \delta_1 y_{21}^3 z_{21}-4 \gamma  y_{21}^2 z_{21}^2+1\right)},
\\
\\
-\dfrac{4 y_{22} \left(2 \beta_1 \delta_1 y_{22}+z_{22}^2 (\beta_1+2 (\delta_1+\alpha  \delta_1 y_{22}))+\alpha  z_{22}^4\right)}{x z_{22}^2 \left(z_{22}^2 \left(4 \gamma  y_{22}^2-1\right)-4 \delta_1 y_{22}\right)}.
\end{cases}
$$
Since $\mathcal{L}_2^{p}$ in these charts is given by $z_{21}=0$ and $y_{22}=0$, the statement follows.
\end{proof}

\begin{lemma}\label{lemma:L1,L8}
For each compact subset $K$ of $\mathcal{L}_1^p\setminus(\mathcal{L}_3^p\cup\mathcal{L}_4^p)$ and
$\mathcal{L}_8^p\setminus(\mathcal{L}_5^p\cup\mathcal{L}_9^p)$, there exists a neighbourhood $V$ of $K$ and a constant $C>0$ such that
$$
\left|
x\frac{E'}{E}
\right|<C
\quad\text{in } V \text{ for all } x\neq0.
$$
\end{lemma}
\begin{proof}
Near $\mathcal{L}_1^p$, in the coordinate charts $(y_{11},z_{11})$ and $(y_{12},z_{12})$, we have:
$$
x\frac{E'}{E}\sim
\begin{cases}
\dfrac{8 \gamma  y_{11}^2+4 \alpha  y_{11}-2}{1-4 \gamma  y_{11}^2},
\\
\\
\dfrac{8 \gamma -2 z_{12}^2+4 \alpha  z_{12}}{z_{12}^2-4 \gamma }.
\end{cases}
$$
Since $\mathcal{L}_3^p$, $\mathcal{L}_4^p$ are given by $y_{11}=\pm1/\sqrt{4\gamma}$ and $z_{12}=\pm\sqrt{4\gamma}$, the expression is bounded for a compact subset $K$ of $\mathcal{L}_1^p\setminus(\mathcal{L}_3^p\cup\mathcal{L}_4^p)$.

Near $\mathcal{L}_8^p$, in the coordinate charts $(y_{81},z_{81})$ and $(y_{82},z_{82})$, we have:
$$
x\frac{E'}{E}\sim
\begin{cases}
\dfrac{2 \left(4 \beta_1 \delta_1 y_{81}^2+2 \beta_1 y_{81}-1\right)}{4 \delta_1 y_{81}+1},
\\
\\
\dfrac{2 \left(4 \beta_1 \delta_1-z_{82}^2+2 \beta_1 z_{82}\right)}{z_{82} (4 \delta_1+z_{82})}.
\end{cases}
$$
In the $(y_{81},z_{81})$ chart, the line $\mathcal{L}_5^p$ is not
visible, and furthermore, the line $\mathcal{L}_9^p$ is projected onto point 
$(-\frac{1}{4\delta_1},0)$.
In the $(y_{82},z_{82})$ chart, the lines $\mathcal{L}_5^p$ iand $\mathcal{L}_9^p$ are projected onto points $(0,0)$ and $(0,-4\delta_1)$.
Thus the expression $x\,E'/E$ will be limited for compact subsets $K$ as desired.
\end{proof}

\begin{lemma}\label{lemma:d}
There exists a continuous complex valued function $d$ on a neighbourhood of the infinity set $\mathcal{I}$ in the Okamoto space, such that:
$$
d=\begin{cases}
1/E & \text{in a neighbourhood of } \mathcal{I}\setminus(\mathcal{L}_3^p\cup\mathcal{L}_4^p\cup\mathcal{L}_5^p\cup\mathcal{L}_9^p),
\\
-J_{61}/\sqrt{\gamma} & \text{in a neighbourhood of } \mathcal{L}_3^p\setminus\mathcal{L}_1^p,
\\
J_{71}/\sqrt{\gamma} & \text{in a neighbourhood of } \mathcal{L}_4^p\setminus\mathcal{L}_1^p,
\\
J_{111}/\delta_1 & \text{in a neighbourhood of } \mathcal{L}_5^p\setminus\mathcal{L}_8^p,
\\
J_{102}/(2\delta_1) & \text{in a neighbourhood of } \mathcal{L}_9^p\setminus\mathcal{L}_8^p.
\end{cases}
$$
\end{lemma}
\begin{proof}
As $1/E$ is well defined in a neighbourhood of
$\mathcal{I}\setminus(\mathcal{L}_3^p\cup\mathcal{L}_4^p\cup\mathcal{L}_5^p\cup\mathcal{L}_9^p)$,
we take 
the first statement as a definition of $d$ there.  The set $\mathcal{L}_3^p\setminus\mathcal{L}_1^p$ is defined by $y_{61}=0$ in the chart $(y_{61},z_{61})$, see Section \ref{sec:b6-blow}.
As we approach $\mathcal{L}_3^p$, we have:
$$
EJ_{61}\sim-\sqrt{\gamma }-\frac{2 \alpha }{x z_{61}}.
$$

The set $\mathcal{L}_5^p\setminus\mathcal{L}_8^p$ is defined by $Y_{111}=0$ in the chart $(Y_{111},Z_{111})$, which is obtained after blowing down the preimage of $\mathcal{L}_0$ in $(y_{111},z_{111})$, see Section \ref{sec:blowdownL0}.
As we approach $\mathcal{L}_5^p$, we have:
$$
EJ_{111}\sim
\delta_1-\frac{2 \beta_1}{x Z_{111}}.
$$

The set $\mathcal{L}_9^p\setminus\mathcal{L}_8^p$ is defined by $z_{102}=0$ in the chart $(y_{102},z_{102})$, see Section \ref{sec:b10-blow}.
As we approach $\mathcal{L}_9^p$, we have:
$$
EJ_{102}\sim2 \delta_1 -\frac{\beta_1+4 \delta_1}{8 \delta_1^2 x y_{102}}.
$$
These estimates lead us to the desired results.\end{proof}

\begin{lemma}[Behaviour near $\mathcal{L}_3^p\setminus\mathcal{L}_1^p$]
\label{lemma:L3-L1}
If a solution at the complex time $x$ is sufficiently close to $\mathcal{L}_3^p\setminus\mathcal{L}_1^p$, there there exists a unique $\xi\in\mathbb{C}$ such that:
\begin{enumerate}
\item
$z_{61}(\xi)=0$, i.e., $(y_{61}(\xi),z_{61}(\xi))\in\mathcal{L}_6(\xi)^p$; and
\item
$|x-\xi|=O(|d(x)||z_{61}(x)|)$ for small $d(x)$ and bounded $|z_{61}(x)|$.
\end{enumerate}
In other words, the solution has a pole at $x=\xi$.

For large $R_3>0$, consider the set
$\{x\in\mathbb{C} \mid |z_{61}(x)|\le R_3 \}$.
Then, its connected component containing $\xi$ is an approximate disk $D_3$ with centre $\xi$ and radius $|d(\xi)|R_3$,
and $x\mapsto z_{61}(x)$ is a complex analytic diffeomorphism from that approximate disk onto
$\{ y\in\mathbb{C} \mid |y|\le R_3\}$.
\end{lemma}
\begin{proof}
For the study of solutions near $\mathcal{L}_3^p\setminus\mathcal{L}_1^p$, we use the coordinates $(y_{61},z_{61})$.
In this chart, the line $\mathcal{L}_3^p\setminus\mathcal{L}_1^p$ is
given by the equation $y_{61}=0$ and parametrized  by
$z_{61}\in\mathbb{C}$, while line $\mathcal{L}_6^p$ is given by
$z_{61}=0$ and parametrized by $y_{61}$, see Section
\ref{sec:b6-blow}. Note that we also have $y=1/(y_{61}z_{61})$.

Asymptotically, for $y_{61}\to0$ and bounded $z_{61}$, $1/x$, we have:
\begin{subequations}
\begin{align}
z_{61}' &\sim-\frac{\sqrt{\gamma}}{y_{61}}
,\label{eq:z61'}
\\
J_{61} &=-y_{61}
,\label{eq:J61}
\\
\frac{J_{61}'}{J_{61}} &\sim -\frac{z_{61}}{2} +\frac1x\cdot\left(2- \frac{\alpha }{\sqrt{\gamma } }\right)
,\label{eq:J61'}
\\
EJ_{61} & \sim -\sqrt{\gamma }-\frac{2 \alpha }{x z_{61}}
.\label{eq:EJ61}
\end{align}
\end{subequations}
Integrating (\ref{eq:J61'}) from $\xi$ to $x$, we get:
\begin{gather*}
J_{61}(x)=J_{61}(\xi)e^{K(x-\xi)}\left(\frac{x}{\xi}\right)^{2 -\alpha/\sqrt{\gamma } }(1+o(1)),
\end{gather*}
with $K=- z_{61}(\tilde\xi)/2$, and $\tilde\xi$ on the integration path.

Because of (\ref{eq:J61}), $y_{61}$ is approximately equal to a small
constant, and from (\ref{eq:z61'}) it follows that
$$
z_{61}\sim z_{61}(\xi)-\frac{\sqrt{\gamma}}{y_{61}(\xi)}(x-\xi).
$$
Thus, if $x$ runs over an approximate disk $D$ centered at $\xi$ with radius $|\gamma^{-1/2}y_{61}|R$, then $z_{61}$ fills an approximate disk centered at $z_{61}(\xi)$ with radius $R$.
Therefore, if $|y_{61}(\xi)|\ll|\xi|$, the solution has the following properties for $x\in D$:
$$
\frac{y_{61}(x)}{z_{61}(x)}\sim1,
$$
and $z_{61}$ is a complex analytic diffeomorphism from $D$ onto an approximate disk with centre $z_{61}(\xi)$ and radius $R$.
If $R$ is sufficiently large, we will have $0\in z_{61}(D)$; i.e.,
the solution $y$ of the Painlev\'e equation will have a pole at a unique point in $D$.

Now, it is possible to take $\xi$ to be the pole point.
For $|x-\xi| \ll |\xi|$, we have
$$
\frac{d(x)}{d(\xi)}\sim1,
$$
and since $d=-J_{61}/\sqrt\gamma$ we obtain 
$$
-\frac{J_{61}(x)}{\sqrt{\gamma} d(\xi)}
=
\frac{y_{61}(x)}{\sqrt{\gamma}d(\xi)}\sim1,
$$
and we get
$$
z_{61}
\sim 
-\frac{\sqrt{\gamma}}{y_{61}(\xi)}(x-\xi)
\sim
-\frac{x-\xi}{d(\xi)}.
$$

Let $R_3$ be a large positive real number.
Then the equation $|z_{61}(x)|=R_3$ corresponds to $|x-\xi|\sim|d(\xi)|R_3$, 
which is still small compared to $|\xi|$ if $|d(\xi)|$ is sufficiently small.
Denote by $D_3$ the connected component of the set of all $x\in\mathbb{C}$ such that
$\{x\mid |z_{61}(x)|\le R_3\}$ is an approximate disk with centre $\xi$ and radius $|d(\xi)|R_3$.
More precisely, $z_{61}$ is a complex analytic diffeomorphism from $D_3$ onto 
$\{y\in\mathbb{C} \mid |y|\le R_3\}$, and
$$
\frac{d(x)}{d(\xi)}\sim1
\quad
\text{for all}
\quad
x\in D_3.
$$
It follows from Equation (\ref{eq:EJ61}) that the function $E(x)$ has a simple
pole at $x=\xi$. But we can estimate the size of the domain where it
becomes bounded and approximately constant.
From Equation (\ref{eq:EJ61}), we have
$$
E(x)J_{61}(x)\sim -\sqrt{\gamma}
\quad
\text{when}
\quad
1\gg
\frac{1}{|xz_{61}(x)|}
\sim
\left| \frac{y_{61}(\xi)}{\sqrt{\gamma}  \xi(x-\xi)}  \right|
\sim
\left| \frac{d(\xi)}{\sqrt{\gamma}  \xi(x-\xi)}  \right|,
$$
that is, when $|x-\xi|\gg\frac{|d(\xi)|}{|\xi|}$.

Since $R_3\gg 1/|\xi|$, the approximate radius of $D_3$ is given by
$$
|d(\xi)|R_3\gg \frac{|d(\xi)|}{|\xi|}.
$$
Thus $E(x)J_{61}(x)\sim-\sqrt{\gamma}$ for $x\in D_3\setminus D_3'$, where $D_3'$ is a disk centered at $\xi$ with small radius compared to the radius of $D_3$.
\end{proof}

\begin{lemma}[Behaviour near $\mathcal{L}_1^p\setminus\mathcal{L}_2^p$]
\label{lemma:L1-L2}
For large finite $R_1>0$, consider the set of all $x\in\mathbb{C}$ such that the solution at complex time $x$ is close to 
$\mathcal{L}_1^p\setminus\mathcal{L}_2^p$, with $|z_{31}(x)|\le R_1$, but not close to $\mathcal{L}_3^p\cup\mathcal{L}_4^p$.
Then this set is the complement of $D_3$ in and approximate disc $D_1$ with centre at $\xi$ and radius $\sim\sqrt{|d(\xi)/(2\gamma)|R_1}$ of an approximate disk with centre at the origin and small radius $\sim4|\sqrt{\gamma}d(\xi)|R_3^2$, where
$z_{31}(x)\sim4\gamma^{3/2}(x-x_0)^2/d(\xi)$.
\end{lemma}
\begin{proof}
The set $\mathcal{L}_1^p\setminus\mathcal{L}_2^p$ is visible in the chart $(y_{31},z_{31})$, where it is given by the equation $y_{31}=0$ and parametrized by $z_{31}$, see Section \ref{sec:b3-blow}.
In that chart, the line $\mathcal{L}_3^p$ is given by $z_{31}=0$.
The projection of line $\mathcal{L}_4^p$ is the point $(y_{31},z_{31})=(0,-4\sqrt{\gamma})$.

For $y_{31}\to0$, bounded $1/x$, and  $z_{31}$ bounded and bounded
away from $-4\sqrt{\gamma}$, we have
\begin{subequations}
\begin{align}
y_{31}' &\sim\frac{\sqrt{\gamma}}{z_{31}}
,\label{eq:y31'}
\\
z_{31}' &\sim-\frac{1}{2y_{31}}(z_{31}+4\sqrt{\gamma})
,\label{eq:z31'}
\\
J_{31}&=-y_{31}^2z_{31}
,\label{eq:J31}
\\
EJ_{31} & \sim \sqrt{\gamma}-\frac{z_{31}}{4}
,\label{eq:EJ31}
\\
\frac{E'}{E} &\sim -\frac2x
+\frac{2 \alpha}{x z_{31}}
+\frac{2 \alpha}{x  \left(4 \sqrt{\gamma }+z_{31}\right)}.
\label{eq:E'31}
\end{align}
\end{subequations}
From (\ref{eq:E'31}) and (\ref{eq:y31'}), we get:
$$
\frac{E'}{E}
\sim
-\frac2x+\frac{2\alpha}{x\sqrt{\gamma}}y_{31}'+\frac{2 \alpha}{x  \left(4 \sqrt{\gamma }+z_{31}\right)},
$$
and then integrating from $x_0$ to $x_1$:
\begin{align*}
\log\frac{E(x_1)}{E(x_0)}
\sim
&\Bigl(-2+\frac{2\alpha}{4\sqrt{\gamma}+z_{31}(\tilde x)}\Bigr)\log\frac{x_1}{x_0}\\
&\quad +
\frac{2\alpha}{\sqrt{\gamma}}
\Bigl(
\frac{y_{31}(x_1)}{x_1}-\frac{y_{31}(x_0)}{x_0}+\int_{x_0}^{x_1}\frac{y_{41}(x)}{x^2}dx
\Bigr),
\end{align*}
with $\tilde{x}$ being on the integration path.
Therefore $E(x_1)/ E(x_0)\sim1$ if for all $x$ on the segment from $x_0$ to $x_1$ we have
$|x-x_0|\ll|x_0|$, $|y_{31}(x)|\ll|x_0|$, and
$1/(4\sqrt{\gamma}+z_{31}(x))$ is bounded.
We choose $x_0$ on the boundary of $D_3$ from Lemma \ref{lemma:L3-L1}.
Then we have:
$$
\frac{d(\xi)}{d(x_0)}\sim E(x_0)d(\xi)\sim -E(x_0)\frac{J_{61}(x_0)}{\sqrt{\gamma}}\sim1
\quad
\text{and}
\quad
|z_{61}(x_0)|=R_3,
$$
which implies that
$$
|y_{31}|
=
\left|
\frac{1}{z_{61}+2\alpha/(x\sqrt{\gamma})}
\right|
\sim
\frac1{R_3}
\ll
1.
$$
Furthermore, Equations (\ref{eq:J31}) and (\ref{eq:EJ31}) imply that
$$
|z_{31}(z_0)|=\frac{|J_{31}(z_0)|}{|y_{31}(z_0)|^2}
\sim
\frac{4|\sqrt{\gamma}d(\xi)|R_3^2}{|1-d(\xi)R_3^2|},
$$
which is small when $|d(\xi)|$ is sufficiently small,
therefore
$$
|z_{31}(z_0)|
\sim
4|\sqrt{\gamma}d(\xi)|R_3^2.
$$

Since $D_3$ is an approximate disk with centre $\xi$ and small radius approximately equal to $d(\xi)R_3$,
and $R_3\gg|\xi|^{-1}$, we have that $|z_{61}(x)|\ge R_3\gg1$.
Writing $z=\xi+r(x_0-\xi)$, where $r\ge1$, we have $|y_{31}(x)|\ll1$ and
$$
\frac{|x-x_0|}{|x_0|}
=
(r-1)\left|1-\frac{\xi}{x_0}\right|
\ll1
\quad
\text{if}
\quad
r-1\ll\frac{1}{|1-\xi/x_0|}.
$$
Then equations (\ref{eq:J31}), (\ref{eq:EJ31}), and $E\sim d(\xi)^{-1}$ yield
$$
y_{31}^{-1}
\sim
\left(
-\frac{z_{31}}{\sqrt{\gamma}d(\xi)}
\right)^{1/2},
$$
which in combination with (\ref{eq:z31'}) leads to
$$
\frac{z_{31}'}{z_{31}^{1/2}(z_{31}+4\sqrt{\gamma})}
\sim 
-\frac12(-\sqrt{\gamma} d(\xi))^{-1/2}
.
$$
Integrating we get
$$
\arctan(z_{31}^{1/2} (4\sqrt{\gamma})^{-1/2})\mid_{x_0}^{x}
\sim
-\sqrt{\gamma}d(\xi)^{-1/2}(x-x_0),
$$
and therefore
$$
z_{31}(x)
\sim
\frac{4\gamma^{3/2}}{d(\xi)}(x-x_0)^2
\quad
\text{if}
\quad
|x-x_0|\gg\sqrt{|z_{31}(x_0)d(\xi)|}.
$$
For large finite $R_1>0$, the equation $|z_{31}|=R_1$ corresponds to $|x-x_0|\sim\sqrt{|d(\xi)/(2\gamma)|R_1}$,
which is still small compared to $|x_0|\sim|\xi|$, 
and therefore $|x-\xi|\le|x-x_0|+|x_0-\xi|\ll|\xi|$.
This proves the statement of the lemma.
\end{proof}

\begin{lemma}[Behaviour near $\mathcal{L}_9^p\setminus\mathcal{L}_8^p$]                 
\label{lemma:L9-L8}
If a solution at the complex time $x$ is sufficiently close to
$\mathcal{L}_9^p\setminus\mathcal{L}_8^p$, there there exists a unique
$\xi\in\mathbb{C}$ such that we have:
\begin{enumerate}
\item
$y_{102}(\xi)=0$, i.e., $(y_{102}(\xi),z_{102}(\xi))\in\mathcal{L}_{10}(\xi)$; and
\item
$|x-\xi|=O(|d(x)||y_{102}(x)|)$ for small $d(x)$ and bounded $|y_{102}(x)|$.
\end{enumerate}
In other words, the solution has a pole at $x=\xi$.

For large $R_9>0$, consider the set
$\{x\in\mathbb{C} \mid |y_{102}|\le R_9 \}$.
Then, its connected component containing $\xi$ is an approximate disk $D_9$ with centre $\xi$ and radius $|d(\xi)|R_9$,
and $x\mapsto y_{102}(x)$ is a complex analytic diffeomorphism from that approximate disk onto
$\{ y\in\mathbb{C} \mid |y|\le R_9\}$.
\end{lemma}
\begin{proof}
For the study of solutions near $\mathcal{L}_9^p\setminus\mathcal{L}_8^p$, we use the coordinates $(y_{102},z_{102})$.
In this chart, the line $\mathcal{L}_9^p\setminus\mathcal{L}_8^p$ is given by the equation $z_{102}=0$ and parametrized by $y_{102}\in\mathbb{C}$, while line $\mathcal{L}_{10}^p$ is given by $y_{102}=0$ and paremetrised by $z_{102}$, see Section \ref{sec:b10-blow}.

Asymptotically, for $z_{102}\to0$ and bounded $y_{102}$, $1/x$, we have:
\begin{subequations}
\begin{align}
y_{102}' &\sim -\frac{\delta_1}{z_{102}}
,\label{eq:y102'}
\\
J_{102} &\sim \frac{z_{102}}{16\delta_1 ^2}
,\label{eq:J102}
\\
\frac{J_{102}'}{J_{102}} &\sim \frac1x\left(4+\frac{\beta_1}{\delta_1}\right),\label{eq:J102'}
\\
EJ_{102} & \sim \delta_1 -\frac{\beta_1+4 \delta_1}{8 \delta_1^2 x y_{102}} .\label{eq:EJ102}
\end{align}
\end{subequations}
Integrating (\ref{eq:J102'}) from $\xi$ to $x$, we get:
\begin{gather*}
J_{102}(x)=J_{102}(\xi)\left(\frac{x}{\xi}\right)^{4+\beta_1/\delta_1 }(1+o(1)).
\end{gather*}
Because of Equation (\ref{eq:J102}), $z_{102}$ is approximately equal
to a small constant, and from (\ref{eq:y102'}) it follows that
$$
y_{102}\sim y_{102}(\xi)-\delta_1\frac{x-\xi}{z_{102}(\xi)}.
$$
Thus, if $x$ runs over an approximate disk $D$ centered at $\xi$ with radius $|z_{102}|R/|\delta_1|$, then $y_{102}$ fills an approximate disk centered at $y_{102}(\xi)$ with radius $R$.
Therefore, if $|z_{102}(\xi)|\ll|\xi|$, the solution has the following properties for $x\in D$:
$$
\frac{z_{102}(x)}{y_{102}(x)}\sim1,
$$
and $y_{102}$ is a complex analytic diffeomorphism from $D$ onto an approximate disk with centre $y_{102}(\xi)$ and radius $R$.
If $R$ is sufficiently large, we will have $0\in y_{102}(D)$; i.e.,  the solution of the Painlev\'e equation will have a pole at a unique point in $D$.

Now, it is possible to take $\xi$ to be the pole point.
For $|x-\xi| \ll |\xi|$, we have
$$
\frac{d(x)}{d(\xi)}\sim1,
$$
that is:
$$
\frac{J_{102}(z)}{2\delta_1 d(\xi)}\sim\frac{z_{102}(x)}{32\delta_1^3d(\xi)}\sim1,
$$
and
$$
y_{102}(x)\sim-\delta_1\frac{x-\xi}{z_{102}(\xi)}\sim-\frac{x-\xi}{32\delta_1^2 d(\xi)}.
$$

Let $R_9$ be a large positive real number.
Then the equation $|y_{102}(x)|=R_9$ corresponds to $|x-\xi|\sim32|\delta_1^2d(\xi)|R_9$, 
which is still small compared to $|\xi|$ if $|d(\xi)|$ is sufficiently small.
Denote by $D_9$ the connected component of the set of all $x\in\mathbb{C}$ such that
$\{x\mid |y_{102}(x)|\le R_9\}$ is an approximate disk with centre $\xi$ and radius $32|\delta_1^2d(\xi)|R_9$.
More precisely, $y_{102}$ is a complex analytic diffeomorphism from $D_9$ onto 
$\{y\in\mathbb{C} \mid |y|\le R_9\}$, and
$$
\frac{d(x)}{d(\xi)}\sim1
\quad
\text{for all}
\quad
x\in D_9.
$$
The function $E(x)$ has a simple pole at $x=\xi$.
From (\ref{eq:EJ102}), we have
$$
E(x)J_{102}(x)\sim\delta_1
\quad
\text{when}
\quad
1\gg
\frac{1}{|xy_{102}(x)|}
\sim
\left| \frac{z_{102}(\xi)}{\delta_1  \xi(x-\xi)}  \right|
\sim
\frac{32|\delta_1^2 d(\xi)|}{|\xi(x-\xi)|},
$$
that is, when $|x-\xi|\gg\frac{|d(\xi)|}{|\xi|}$.

Since $R_9\gg 1/|\xi|$, the approximate radius of $D_9$ is given by
$$
|d(\xi)|R_9\gg \frac{|d(\xi)|}{|\xi|}.
$$
Thus $E(x)J_{102}(x)\sim\delta_1$ for $x\in D_9\setminus D_9'$, where $D_9'$ is a disk centered at $\xi$ with small radius compared to the radius of $D_9$.
\end{proof}

\begin{lemma}[Behaviour near $\mathcal{L}_8^p\setminus(\mathcal{L}_2^p\cup\mathcal{L}_5^p)$]                 
\label{lemma:L8-L9}
For large finite $R_8>0$, consider the set of all $x\in\mathbb{C}$ such that the solution at complex time $x$ is close to 
$\mathcal{L}_8^p\setminus\mathcal{L}_2^p$, with $|y_{92}(x)|\le R_8$,
but not close to $\mathcal{L}_9^p$.
Then that set is the complement of $D_9$ in an approximate disk $D_8$ with centre at $\xi$ and radius
$\sim4\sqrt{|\delta_1d(\xi)|R_8}$.
More precisely, $x\mapsto y_{92}(x)$ defines a $2$-fold covering from the annular domain $D_8\setminus D_9$ onto the complement in
$\{
x\in\mathbb{C}\mid|x|\le R_8
\}$
of an approximate disk with centre at the origin and small radius $\sim16|\delta_1 d(\xi)| R_9^2$,
where $y_{92}(x)\sim(x-\xi)^2/(16\delta_1d(\xi))$.
\end{lemma}
\begin{proof}
The line $\mathcal{L}_8^p\setminus\mathcal{L}_5^p$ is visible in the coordinate system $(y_{92},z_{92})$, where it is given by $z_{92}=0$ and parametrized by $y_{92}$, see Section \ref{sec:b9-blow}.
In that chart, the line $\mathcal{L}_9^p$ is given by the equation $y_{92}=0$ and parametrized by $z_{92}$.
On the other hand, the line $\mathcal{L}_2^p$ is given by $y_{92}=1/(4\delta_1)$.
For $z_{92}\to0$, $y_{92}$ bounded and bounded away from $1/(4\delta_1)$, and bounded $1/x$, we have:
\begin{subequations}
\begin{align}
y_{92}' &\sim -\frac{2\delta_1}{z_{92}}
,\label{eq:y92'}
\\
z_{92}' &\sim \frac{\delta_1}{y_{92}}+\frac{8\delta_1^2}{4 \delta_1 y_{92}-1}
,\label{eq:z92'}
\\
J_{92} &= \frac{y_{92} z_{92}^2 (1-4 \delta_1 y_{92})^2}{16 \delta_1^2}
,\label{eq:J92}
\\
EJ_{92} &\sim \delta_1,\label{eq:EJ92}
\\
\frac{E'}{E} & \sim \frac{2\beta_1 y_{92}}{x}
+
\frac{\beta_1 - 4 \delta_1}{8\delta_1^2 x y_{92}}
.\label{eq:E'/E92}
\end{align}
\end{subequations}

From (\ref{eq:E'/E92}) and (\ref{eq:z92'}) we get:
$$
\frac{E'}{E}\sim \frac{2\beta_1 y_{92}}{x}
+
\frac{\beta_1-4\delta_1}{8\delta_1^3x}z_{92}'
-
\frac{\beta_1-4\delta_1}{\delta_1x(4 \delta_1 y_{92}-1)}.
$$
Integrating from $x_0$ to $x_1$, we obtain:
$$
\begin{aligned}
\log\frac{E(x_1)}{E(x_0)}
\sim&
2\beta_1\int_{x_0}^{x_1}\frac{y_{92}}{x}dz
+
\frac{\beta_1-4\delta_1}{8\delta_1^3}
\left(
\frac{z_{92}(x_1)}{x_1}-\frac{z_{92}(x_0)}{x_0}+\int_{x_0}^{x_1}\frac{z_{92}(x)}{x^2}dx
\right)
\\&
-
\frac{\beta_1-4\delta_1}{\delta_1}
\left(
\int_{x_0}^{x_1}\frac{dx}{x(4 \delta_1 y_{92}-1)}
\right).
\end{aligned}
$$
Therefore $E(x_1)/E(x_0)\sim1$ if for all $x$ on the integration path we have $|x-x_0|\ll |x_0|$ and
$| y_{92}(x)|\ll |x_0|$,
$|z_{92}(x)|\ll |x_0|$, 
$1/| 4 \delta_1 y_{92}(\tilde x)-1|\ll |x_0|$.

We choose $x_0$ on the boundary of $D_9$ from Lemma \ref{lemma:L9-L8}.
Then we have
$$
\frac{d(\xi)}{d(x_0)}\sim \frac{J_{102}(\xi)}{2\delta_1}E(x_0)\sim1
\quad
\text{and}
\quad
|y_{102}(x_0)|=R_9.
$$
Since
$$
z_{92}=\frac{1}{y_{91}}=\frac{1}{y_{102}-(\beta_1+4\delta_1)/(8x\delta_1^3)},
$$
we are led to conclude that
$$
|z_{92}|\sim\frac1{R_9}\ll1.
$$
Furthermore, equations (\ref{eq:J92}) and (\ref{eq:EJ92}) imply that:
$$
|y_{92}(1-4\delta_1 y_{92})^2|
=
\frac{16|\delta_1^2 J_{92}|}{|z_{92}^2|}
\sim
16|\delta_1^3 d(\xi)| R_9^2,
$$
which is small when $|d(\xi)|$ is sufficiently small.

Since $D_9$ is an approximate disk with centre $\xi$ and small radius $\sim |d(\xi)|R_9$, and $R_9\gg |\xi|^{-1}$, we have that
$|y_{102}(x)|\ge R_9\gg1$. It follows that 
$$
|y_{92}|\ll1
\quad
\text{if}
\quad
x=\xi+r(x_0-\xi),
\ r\ge 1,
$$
and
$$
\frac{|x-x_0|}{|x_0|}
=
(r-1)\left|1-\frac{\xi}{x_0}\right|\ll1
\quad
\text{if}
\quad
r-1\ll\frac{1}{|1-\frac{\xi}{x_0}|}.
$$

Then equation (\ref{eq:J92}) and the result $J_{92}\sim\delta_1
d(\xi)$ (from Equation \eqref{eq:J92}) yield
$$
z_{92}^{-1}
\sim
\left( \frac{y_{92}(1-4\delta_1 y_{92})^2}{16\delta_1^3 d(\xi)}\right)^{1/2}
\sim
\left(\frac{y_{92}}{16\delta_1^3 d(\xi)}\right)^{1/2},
$$
which in combination with (\ref{eq:y92'}) leads to
$$
\frac{y_{92}'}{y_{92}^{1/2}}\sim-\frac1{2(\delta_1d(\xi))^{1/2}}.
$$
In other words, we have
$$
(y_{92}^{1/2})'\sim-\frac{1}{4\delta^{1/2}d(\xi)^{1/2}},
$$
from which we find after integration
$$
y_{92}^{1/2}\sim y_{92}(x_0)^{1/2}-\frac{x-x_0}{4\delta_1^{1/2}d(\xi)^{1/2}}.
$$
Finally, we conclude that
$$
y_{92}(x)\sim\frac{(x-x_0)^2}{16\delta_1d(\xi)}
\quad
\text{if}
\quad
|x-x_0|\gg|y_{92}(x_0)^{1/2}|.
$$
For large finite $R_8$, the equation $|y_{92}|=R_8$ corresponds to $|x-x_0|\sim4\sqrt{|\delta_1d(\xi)|R_8}$,
which is still small compared to $|x_0|\sim|\xi|$, and therefore
$|x-\xi|\le|x-x_0|+|x_0-\xi|\ll|\xi|$.
This proves the lemma.
\end{proof}

\begin{lemma}[Behaviour near $\mathcal{L}_5^p\setminus\mathcal{L}_8^p$]
\label{lemma:L5-L8}
If a solution at the complex time $x$ is sufficiently close to $\mathcal{L}_5^p\setminus\mathcal{L}_8^p$, then there exists a unique $\xi\in\mathbb{C}$ such that:
\begin{enumerate}
\item
$Z_{111}(\xi)=0$, i.e., $(Y_{111}(\xi),Z_{111}(\xi))\in\mathcal{L}_{11}(\xi)^p$; and
\item
$|x-\xi|=O(|d(x)||Z_{111}(x)|)$ for small $d(x)$ and bounded $|Y_{111}(x)|$.
\end{enumerate}
In other words, the solution has a pole at $x=\xi$.

For large $R_{5}>0$, consider the set
$\{x\in\mathbb{C} \mid |Z_{111}|\le R_{5} \}$.
Then, its connected component containing $\xi$ is an approximate disk $D_5$ with centre $\xi$ and radius $|d(\xi)|R_5$,
and $x\mapsto Z_{111}(x)$ is a complex analytic diffeomorphism from that approximate disk onto
$\{ z\in\mathbb{C} \mid |z|\le R_5\}$.
\end{lemma}
\begin{proof}
For the study of solutions near $\mathcal{L}_5^p\setminus\mathcal{L}_8^p$, we use the coordinates $(Y_{111},Z_{111})$.
In this chart, the line $\mathcal{L}_5^p\setminus\mathcal{L}_8^p$ is given by the equation $Y_{111}=0$ and parametrized by $Z_{111}\in\mathbb{C}$, while line $\mathcal{L}_{11}^p$ is given by $Z_{111}=0$ and parametrized by $Y_{111}$, see Section \ref{sec:blowdownL0}.

Asymptotically, for $Y_{111}\to0$ and bounded $Z_{111}$, $1/x$, we have:
\begin{subequations}
\begin{align}
Z_{111}' &\sim \frac{\delta_1}{Y_{111}}
,\label{eq:Z111'}
\\
J_{111} &=Y_{111}
,\label{eq:J111}
\\
\frac{J_{111}'}{J_{111}} &\sim 
\frac{Z_{111}}{2}-\frac{\beta_1}{\delta_1 x},
\label{eq:J111'}
\\
EJ_{111} & \sim \delta_1-\frac{2 \beta_1}{x Z_{111}}
.\label{eq:EJ111}
\end{align}
\end{subequations}
Integrating Equation (\ref{eq:J111'}) from $\xi$ to $x$, we get:
\begin{gather*}
\frac{J_{111}(x)}{J_{111}(\xi)}\sim e^{K(x-\xi)}\left(\frac{x}{\xi}\right)^{ -3 \beta_1/\delta_1 },
\end{gather*}
with $K=Z_{111}(\tilde\xi)/2$, where $\tilde\xi$ lies on the integration path.

Because of Equation (\ref{eq:J111}), $Y_{111}$ is approximately equal to a small constant, and from (\ref{eq:Z111'}) follows that
$$
Z_{111}\sim Z_{111}(\xi)+\delta_1\frac{x-\xi}{Y_{111}(\xi)}.
$$
Thus, if $x$ runs over an approximate disk $D$ centered at $\xi$ with radius $|Y_{111}|R/\delta_1$, then $Z_{111}$ fills an approximate disk centered at $Z_{111}(\xi)$ with radius $R$.
Therefore, if $|Y_{111}|\ll1/|\xi|$, for $x\in D$, the solution satisfies
$$
\frac{Y_{111}(x)}{Y_{111}(\xi)}\sim1,
$$
and $Z_{111}(x)$ is a complex analytic diffeomorphism from $D$ onto an approximate disk with centre $Z_{111}(\xi)$ and radius $R$.
If $R$ is sufficiently large, we will have $0\in Z_{111}(D)$; i.e., the solution of the Painlev\'e equation will have a pole at a point in $D$.

Now, it is possible to take $\xi$ to be the pole point.
For $|x-\xi|\ll|\xi|$, we have
$$
\frac{d(x)}{d(\xi)}\sim1
\quad\text{that is,}\quad
\frac{Y_{111}(x)}{\delta_1 d(\xi)}=\frac{J_{111}(x)}{\delta_1 d(\xi)}\sim1,
$$
and
$$
Z_{111}(x)\sim\delta_1\frac{x-\xi}{Y_{111}}
\sim
\frac{x-\xi}{d(\xi)}.
$$
Let $R_5$ be a large positive number.
Then the equation $|Z_{111}|=R_5$ corresponds to $|x-\xi|\sim|d(\xi)|R_5$, which is still small compared to $\xi$ if $d(\xi)$ is sufficiently small.
Denote by $D_5$ the connected component of the set of all $x\in\mathbb{C}$ such that
$\{ x \mid |Z_{111}|\le R_5 \}$ is an approximate disk with centre $\xi$ and radius $d(\xi)R_5$.

More precisely, $Z_{111}$ is a complex analytic diffeomorphism from $D_5$ onto $\{ Z\in\mathbb{C}\mid|Z|\le R_5\}$,
and
$$
\frac{d(x)}{d(\xi)}\sim1
\quad\text{for all}\quad
x\in D_4.
$$
Because of the expression for $E(x)$ in terms of $Y_{111}$, $Z_{111}$,
it has a simple pole at $x=\xi$.
From (\ref{eq:EJ111}), we have
$$
E(x)J_{111}(x)\sim\delta_1
\quad\text{when}\quad
1\gg\frac1{|xZ_{111}(x)|}
\sim
\left|
\frac{Y_{111}(\xi)}{\delta_1\xi(x-\xi)}
\right|
\sim
\left|
\frac{d(\xi)}{\xi(x-\xi)}
\right|,
$$
that is, when
$$
|x-\xi|\gg\frac{|d(\xi)|}{|\xi|}.
$$
Thus $E(x)J_{111}(x)\sim\delta_1$ for the annular disk $x\in D_4\setminus D_4'$, 
where $D_4'$ is a disk centered at $\xi$ with small radius compared to the radius of $D_4$.
\end{proof}

\begin{lemma}[Behaviour near $\mathcal{L}_8^p\setminus(\mathcal{L}_2^p\cup\mathcal{L}_9^p)$]                 
\label{lemma:L8-L2}
For large finite $R_8>0$, consider the set of all $x\in\mathbb{C}$ such that the solution at complex time $x$ is close to 
$\mathcal{L}_8^p\setminus\mathcal{L}_2^p$, with $|y_{92}(x)|\le R_8$,
but not close to $\mathcal{L}_5^p$.
Then that set is the complement of $D_5$ in an approximate disk $D_8$ with centre at $\xi$ and radius
$\sim4\sqrt{|\delta_1d(\xi)|R_8}$.
More precisely, $x\mapsto y_{92}(x)$ defines a $2$-fold covering from the annular domain $D_8\setminus D_5$ onto the complement in
$\{
x\in\mathbb{C}\mid|x|\le R_8
\}$
of an approximate disk with centre at the origin and small radius $\sim16|\delta_1 d(\xi)| R_5^2$,
where $y_{92}(x)\sim(x-\xi)^2/(16\delta_1d(\xi))$.
\end{lemma}
\begin{proof}
We consider the coordinate system $(y_{92},z_{92})$, as in the proof of Lemma \ref{lemma:L8-L9}.
For $z_{92}\to0$, $y_{92}$ bounded and bounded away from $1/(4\delta_1)$ and $0$, and bounded $1/x$, the estimates (\ref{eq:y92'}--\ref{eq:E'/E92}) hold. 

We choose $x_0$ on the boundary of $D_5$ from Lemma \ref{lemma:L5-L8}.
Then we have
$$
\frac{d(\xi)}{d(x_0)}\sim \frac{J_{111}(\xi)}{\delta_1}E(x_0)\sim1
\quad
\text{and}
\quad
|Z_{111}(x_0)|=R_5.
$$
From the results of Sections \ref{sec:blowdownL0}, \ref{sec:b11-blow},
\ref{sec:b5-blow}, \ref{sec:b8-blow}, we have
$$
Z_{111}
=z_{111}
=z_{52}+\frac{2\beta_1}{x\delta_1}
=\frac{1}{y_{51}}+\frac{2\beta_1}{x\delta_1}
=\frac{1}{y_{81}z_{81}}+\frac{2\beta_1}{x\delta_1}
=\frac{1}{y_{92}z_{92}(y_{92}-\frac{1}{4\delta_1})}+\frac{2\beta_1}{x\delta_1},
$$
which implies that $|z_{92}|R_5$ is bounded and limited away from $0$.
In particular, we have $|z_{92}|\ll1$.

Following the same steps and calculations as in Lemma
\ref{lemma:L8-L9}, we get the desired conclusion.
\end{proof}

\begin{theorem}\label{th:estimates}
Let $\epsilon_1$, $\epsilon_2$, $\epsilon_3$ be given such that $\epsilon_1>0$, $0<\epsilon_2<2$, $0<\epsilon_3<1$.
Then there exists $\eta>0$ such that if $|x_0|>\epsilon_1$ and $|d(x_0)|<\eta$, then
$$
\rho=\sup\{ r>|x_0| \text{ such that } |d(x)|<\eta \text{ whenever } |x_0|\le|x|\le r \},
$$
satisfies:
\begin{enumerate}
\item
$\eta\ge|d(x_0)|\left(\dfrac{\rho}{|x_0|}\right)^{2-\epsilon_2}(1-\epsilon_3)$;
\item
if $|x_0|\le|x|\le\rho$, then $d(x)=d(x_0)\left(\dfrac{x}{x_0}\right)^{2+\epsilon_2(x)}(1+\epsilon_3(x))$;
\item
if $|x|\ge\rho$, then $d(x)\ge\eta(1-\epsilon_3)$.
\end{enumerate}
\end{theorem}

\begin{proof}
Suppose a solution of the system (\ref{eq:p3-system}) is close to $\mathcal{L}_2^p$ at times $x_0$ and $x_1$.
It follows from Lemmas \ref{lemma:L3-L1}--\ref{lemma:L8-L9} that for every solution close to $\mathcal{I}$, the set of complex times $x$ such that the solution is not close to $\mathcal{L}_2^p$ is the union of approximate disks of radii $\sim\sqrt{|d|}$.
Hence if the solution is near $\mathcal{I}$ for all complex times $x$ such that $|x_0|\le|x|\le|x_1|$, then there exists a path $\Gamma$ from $x_0$ to $x_1$ such that the solution is close to $\mathcal{L}_2^p$ for all $x\in\Gamma$ and $\Gamma$ is $C^1$-close to the path:
$t\mapsto x_1^t z_0^{1-t}$, $t\in[0,1]$.

Then Lemma \ref{lemma:L2} implies (after integration) that
$$
E(x)=E(x_0)\left(\frac{x}{x_0}\right)^{-2+o(1)}(1+o(1)),
$$
and
\begin{equation}\label{eq:d-estimate}
d(z)=d(z_0)\left(\frac{x}{x_0}\right)^{2+o(1)}(1+o(1)).
\end{equation}
Therefore, from Lemmas \ref{lemma:L3-L1}--\ref{lemma:L8-L9} then we have that, as long as the solution is close to $\mathcal{I}$, as it moves into a neighbourhood of 
$\mathcal{L}_1^p\setminus(\mathcal{L}_3^p\cup\mathcal{L}_4^p)$ 
and $\mathcal{L}_8^p\setminus(\mathcal{L}_5^p\cup\mathcal{L}_9^p)$, 
the ratio of $d$ remains close to $1$.

For the first statement of the theorem, we have:
$$
\eta>|d(x)|\ge|d(x_0)|\left|\frac{x}{x_0}\right|^{2-\epsilon_2}(1-\epsilon_3),
$$
and so
$$
\eta\ge\sup_{ \{x\mid|d(x)|<\eta\} } |d(x_0)|\left|\frac{x}{x_0}\right|^{2-\epsilon_2}(1-\epsilon_3).
$$
The second statement follows from (\ref{eq:d-estimate}), while the third one follows by the assumption on $x$.
\end{proof}

In the following corollary, we summarise the results obtained in this section.

\begin{corollary}\label{cor:infinity}
No solution of $(\ref{eq:p3-system})$ that starts in the space of initial values intersects $\mathcal{I}$.
A solution that approaches $\mathcal{I}$ will stay in its vicinity
only for a limited range of the independent variable $x$.
Moreover, if a solution is sufficiently close to $\mathcal{I}$ at a point $x$, then it will have a pole in a neighbourhood of $x$. 
\end{corollary}

\begin{proof}
The first two statements follow from Theorem \ref{th:estimates}, and the last one from Lemmas \ref{lemma:L2}--\ref{lemma:L8-L2}.
\end{proof}

%% file: 5-limit/limit.tex
In this section, we define the complex limit set of the solutions, when $x\to\infty$, and consider its properties. These properties enable to prove that each non-rational solution of the system (\ref{eq:p3-system}) has an infinite number of poles and zeroes.

Our definition of the limit set extends the standard concept of limit sets in dynamical systems to complex-valued solutions.
\begin{definition}
Let $(y(x),z(x))$ be a solution of (\ref{eq:p3-system}). \emph{The limit set} $\Omega_{(y,z)}$ of $(y(x),z(x))$ is the set of all $s\in\mathcal{F}_{\infty}\setminus\mathcal{I}_{\infty}$ such that there exists a sequence $x_n\in\mathbf{C}$ satisfying:
$$
\lim_{n\to\infty}|x_n|=\infty
\quad\text{and}\quad
\lim_{n\to\infty}(y(x_n),z(x_n))=s.
$$
\end{definition}

\begin{theorem}\label{th:limit}
There exists a compact set $K\subset\mathcal{F}(\infty)\setminus\mathcal{I}(\infty)$ such that the limit set $\Omega_{(y,z)}$ of any solution $(y,z)$ of (\ref{eq:p3-system}) is contained in $K$.
Moreover, $\Omega_{(y,z)}$ is a nonempty, compact, and connected set, which is invariant under the flow of the autonomous system (\ref{eq:p3-system-auto}).
\end{theorem}
\begin{proof}
For any positive numbers $\eta$, $r$, let $K_{\eta,r}$ denote the set of all $s\in\mathcal{F}(x)$ such that $|x|\ge r$ and $|d(s)|\ge\eta$.
Since $\mathcal{F}(x)$ is a complex analytic family over $\mathbb{P}^1\setminus\{0\}$ of compact surfaces $\mathcal{F}(x)$, $K_{\eta,r}$ is also compact.
Furthermore, $K_{\eta,r}$ is disjoint from the infinity sets $\mathcal{I}(x)$, $x\in\mathbb{P}^1\setminus\{0\}$, and therefore $K_{\eta,r}$ is a compact subset of the Okamoto's space $\mathcal{O}\setminus\mathcal{F}(\infty)$.
When $r$ grows to infinity, the sets $K_{\eta,r}$ shrink to the compact set
$$
K_{\eta,\infty}=\{s\in\mathcal{F}(\infty)\mid |d(s)|\ge\eta\}
\subset
\mathcal{F}(\infty)\setminus\mathcal{I}(\infty).
$$
It follows from Theorem \ref{th:estimates} that there exists $\eta>0$ such that for every solution $(y,z)$ there exists $r_0>0$ with the following property:
$$
(y(x),z(x))\in K_{\eta, r_0}
\text{ for every }
x
\text{ such that }
|x|\ge r_0.
$$
Hereafter, we take $r\ge r_0$, when it follows that $(y(x),z(x))\in K_{\eta, r}$ whenever $|x|\ge r$.

Let $X_r=\{x\in\mathbb{C}\mid |x|\ge r\}$, and let $\Omega_{(y,z),r}$ denote the closure of the image set $\bigl(y(X_r), z(X_r)\bigr)$ in $\mathcal{O}$.
Since $X_r$ is connected and $(y,z)$ is continuous, $\Omega_{(y,z),r}$ is also connected.
Since $\bigl(y(X_r), z(X_r)\bigr)$ is contained in the compact set $K_{\eta, r}$, its closure $\Omega_{(y,z),r}$ is also contained in $K_{\eta, r}$, 
and therefore $\Omega_{(y,z),r}$ is a nonempty compact and connected subset of $\mathcal{O}\setminus\mathcal{F}(\infty)$.

The intersection of a decreasing sequence of nonempty, compact, and connected sets is a nonempty, compact, and connected. Therefore, as $\Omega_{(y,z),r}$ decreases to $\Omega_{(y,z)}$ as $r$ grows to infinity, it follows that $\Omega_{(y,z)}$ is a nonempty, compact and connected subset of $\mathcal{O}$.
Since $\Omega_{(y,z),r}\subset K_{\eta,r}$, for all $r\ge r_0$, and the sets $K_{\eta,r}$ shrink to the compact subset $K_{\eta,\infty}$ of $\mathcal{F}(\infty)\setminus\mathcal{I}(\infty)$ as $r$ grows to infinity, it follows that $\Omega_{(y,z)}\subset K_{\eta,\infty}$.
This proves the first statement of the theorem with $K=K_{\eta,\infty}$.

Since $\Omega_{(y,z)}$ is the intersection of the decreasing family of compact sets $\Omega_{(y,z),r}$, there exists for every neighbourhood $A$ of $\Omega_{(y,z)}$ in $\mathcal{O}$, an $r>0$ such that $\Omega_{(y,z),r}\subset A$. Hence $(u(x),v(x))\in A$ for every $x\in\mathbb{C}$ such that $|x|\ge r$.
If $\{x_j\}$ is any sequence in $\mathbb{C}\setminus\{0\}$ such that $|x_j|\to\infty$, then the compactness of $K_{\eta,r}$, in combination with $\bigl(y(X_r), z(X_r)\bigr)\subset K_{\eta,r}$, implies that there is a subsequence $j=j(k)\to\infty$ as $k\to\infty$ and an $s\in K_{\eta,r}$, such that:
$$
(y(x_{j(k)}),z(x_{j(k)}))\to s\ \text{as}\ k\to\infty.
$$
It follows, therefore, that $s\in\Omega_{(y,z)}$.

Next, we prove that $\Omega_{(y,z)}$ is invariant under the flow $\Phi^t$ of the autonomous Hamiltonian system.
Let $s\in\Omega_{(y,z)}$ and $x_j$ be a sequence in $\mathbb{C}\setminus\{0\}$ such that $x_j\to\infty$ and $(y(x_j),z(x_j))\to s$.
Since the $x$-dependent vector field of the Painlev\'e system converges in $C^1$ to the vector field of the autonomous Hamiltonian system as $x\to\infty$, it follows from the continuous dependence on initial data and parameters, that the distance between $(y(x_j+t),z(x_j+t))$ and $\Phi^t(y(x_j),z(x_j))$ converges to zero as $j\to\infty$.
Since $\Phi^t(y(x_j),z(x_j))\to\Phi^t(s)$ and $x_j\to\infty$ as $j\to\infty$, it follows that $(y(x_j+t),z(x_j+t))\to\Phi^t(s)$ and $x_j+t\to\infty$ as $j\to\infty$, hence $\Phi^t(s)\in\Omega_{(y,z)}$.
\end{proof}

\begin{proposition}\label{prop:intersections}
Every non-rational solution $(y(x),z(x))$ intersects each of the pole lines $\mathcal{L}_6$, $\mathcal{L}_7$, $\mathcal{L}_{10}$ infinitely many times. 
\end{proposition}

\begin{proof}
First, suppose that a solution $(y(x),z(x))$ intersects the union $\mathcal{L}_6\cup\mathcal{L}_7\cup\mathcal{L}_{10}$ only finitely many times.

According to Theorem \ref{th:limit}, the limit set $\Omega_{(y,z)}$ is a compact set in $\mathcal{F}(\infty)\setminus\mathcal{I}(\infty)$.
If $\Omega_{(y,z)}$ intersects one the three pole lines $\mathcal{L}_6$, $\mathcal{L}_7$, $\mathcal{L}_{10}$ at a point $p$, then there exists arbitrarily large $x$ such that $(y(x),z(x))$ is arbitrarily close to $p$, when the transversality of the vector field to the pole line implies that $(y(\xi),z(\xi))\in\mathcal{L}_6\cup\mathcal{L}_7\cup\mathcal{L}_{10}$ for a unique $\xi$ near $x$.
As this would imply that $(y(x),z(x))$ intersects $\mathcal{L}_6\cup\mathcal{L}_7\cup\mathcal{L}_{10}$ infinitely many times, it follows that $\Omega_{(y,z)}$ is a compact subset of 
$\mathcal{F}(\infty)\setminus(\mathcal{I}({\infty})\cup\mathcal{L}_6({\infty})\cup\mathcal{L}_7({\infty})\cup\mathcal{L}_{10}({\infty}))$.

However, $\mathcal{L}_6({\infty})\cup\mathcal{L}_7({\infty})\cup\mathcal{L}_{10}({\infty})$ is equal to the set of all points in 
$\mathcal{F}(\infty)\setminus\mathcal{I}(\infty)$ which project to the line $\mathcal{L}_{\infty}$, and therefore 
$\mathcal{F}(\infty)\setminus(\mathcal{I}({\infty})\cup\mathcal{L}_6({\infty})\cup\mathcal{L}_7({\infty})\cup\mathcal{L}_{10}({\infty}))$
is the affine $(y,z)$ coordinate chart, of which $\Omega_{(y,z)}$ is a compact subset, which implies that $y(x)$ and $z(x)$ remain bounded for large $|x|$.
It follows from the boundedness of $y$ and $z$ that $y(x)$ and $z(x)$ are equal to holomorphic functions of $1/x$ in a neighbourhood of $x=\infty$, which implies that there are complex numbers $y(\infty)$, $z(\infty)$ which are the limit points of $y(x)$ and $z(x)$ as 
$|x|\to\infty$. 
In other words, $\Omega_{(y,z)}=\{(y(\infty),z(\infty))\}$ contains only one point.
That means that that the solution is analytic at infinity, i.e.,~ it is analytic on the whole of $\mathbb{CP}^1$, thus it must be rational.

Now consider non-rational solutions. Since the limit set $\Omega_{(y,z)}$ is invariant under the autonomous flow, it means that it will contain the whole quartic curve
$
\frac{y^2z^2}{4}+\delta_1 z-\gamma y^2=c,
$
for some constant $c$.
Such a curve contains both the base points $b_3$, $b_4$ on the line $\mathcal{L}_1$, as well as the base point $b_5$ on the line $\mathcal{L}_2$, which are projections of the pole lines $\mathcal{L}_6(\infty)$, $\mathcal{L}_7(\infty)$, $\mathcal{L}_{10}(\infty)$ respectively.
Thus, a non-rational solution will intersect each of these lines infinitely many times.    
\end{proof}

\begin{remark}\label{rem:rational}
The limit set $\Omega_{(y,z)}$ is invariant under the autonomous Hamiltonian system.
If it contains only one point, as we obtained in the proof of Theorem \ref{prop:intersections}, that point must be an equilibrium point of the autonomous Hamiltonian system (\ref{eq:p3-system-auto}).
\end{remark}

\begin{theorem}\label{th:poleszeroes}
Every non-rational solution of the third Painlev\'e equation has infinitely many poles and infinitely many zeros.
\end{theorem}
\begin{proof}
It is enough to prove that a non-rational solution $y$ of (\ref{eq:p3-system}) has infinitely many poles and zeroes.
Notice that at the intersection points with $\mathcal{L}_6$ and $\mathcal{L}_7$, $y$ has a pole and $z$ a zero; at the intersection with $\mathcal{L}_{10}$, $y$ has a zero and $z$ a pole.
Since it is shown in Propostion \ref{prop:intersections} that $(y,z)$ intersects each of the lines  $\mathcal{L}_6$, $\mathcal{L}_7$, $\mathcal{L}_{10}$ infinitely many times, the statement is proved.
\end{proof}

%% file: A-p3resolution/p3resolution.tex
In this section, we give an explicit construction of the space of
initial values for the system (\ref{eq:p3-system}).
This constructions consists of eleven successive blow-ups of points in $\mathbb{CP}^2$.

We use the following notation.
The coordinates in three affine charts of $\mathbb{CP}^2$ are denoted by $(y_{01},z_{01})$,  $(y_{02},z_{02})$, and  $(y_{03},z_{03})$.
The exceptional line obtained in the $n$-th blow-up is covered by two coordinate charts, denoted by $(y_{n1},z_{n1})$ and $(y_{n2},z_{n2})$.

In each of these charts, we write the system (\ref{eq:p3-system}) in the corresponding coordinates and look for \emph{base points} -- the points contained by infinitely many solutions.
We calculate the coordinates of the base points in local coordinates in the following way.
In each chart $(y_{nj},z_{nj})$, the system can be written in the form:
$$
y_{nj}'=\frac{P(y_{nj},z_{nj},x)}{Q(y_{nj},z_{nj},x)},
\quad
z_{nj}'=\frac{R(y_{nj},z_{nj},x)}{S(y_{nj},z_{nj},x)},
$$
for some polynomial expressions $P$, $Q$, $R$, $S$.
The uniqueness of solutions for a given initial value breaks down
whenever $P=Q=0$ or $R=S=0$. So,
solving these equations yields to base points.
We note that, after blowing up, any new base point can appear only on the exceptional line.

We remark that a base point in algebraic geometry is a joint point
contained in all curves from a given pencil.
The solutions of the autonomous system (\ref{eq:p3-system-auto}) are algebraic curves from the pencil (\ref{eq:pencil}), hence the notion of base points of a system of differential equations and base point of a pencil of curves coincide in the autonomous case.

\subsection{Affine charts}\label{sec:affine}

\paragraph{In the first affine chart,} the system has no base points.
The line $\mathcal{L}_0(y=0)$ concides with the set of unachievable points of the system (\ref{eq:p3-system}) in this chart.
\\
\paragraph{In the second affine chart,} we have the following coordinates:
\begin{gather*}
y_{02}=\frac1{y},
\quad
z_{02}=\frac{z}{y},
\\
y=\frac{1}{y_{02}},
\quad
z=\frac{z_{02}}{y_{02}}.
\end{gather*}
The line at the infinity is $\mathcal{L}_{\infty}:y_{02}=0$.
The line $\mathcal{L}_0$ is not visible in this chart.

The Painlev\'e vector field is:
$$
\begin{aligned}
y_{02}'&=-\frac{z_{02}}{2y_{02}}-\delta_1y_{02}^2+\frac{y_{02}}{x},
\\
z_{02}'&=2 \gamma -\delta_1 y_{02}z_{02} - \frac{z_{02}^2}{y_{02}^2}
+\frac{1}{x}(2 \alpha y_{02} + 2 \beta_1 y_{02}^3 + z_{02})
.
\end{aligned}
$$
It follows from the first equation that there is a base point at
$$
b_1\ :\ y_{02}=0,z_{02}=0.
$$
The flow of (\ref{eq:p3-system}) passes through $b_1$ at point $x$ if and only if $y$ has a pole at $x$ and $z$ either is regular at $x$ or has a pole of smaller order than $y$ at $x$.

\paragraph{In the third affine chart,} the coordinates are:
\begin{gather*}
y_{03}=\frac1{z},
\quad
z_{03}=\frac{y}{z},
\\
y=\frac{z_{03}}{y_{03}},
\quad
z=\frac{1}{y_{03}}.
\end{gather*}
The line at the infinity is $\mathcal{L}_{\infty}:y_{03}=0$.
The line $\mathcal{L}_0$ is given by $z_{03}=0$ in this chart.

The vector field is:
$$
\begin{aligned}
y_{03}'&=\frac{z_{03}}{2 y_{03}} - 2\gamma y_{03} z_{03}
-\frac{2}{x}\left(
  \beta_1\frac{ y_{03}^4}{z_{03}^2 }
  + \alpha y_{03}^2 
\right)
,
\\
z_{03}'&=\delta_1y_{03} - 2\gamma z_{03}^2 + \frac{z_{03}^2}{y_{03}^2}
-\frac{1}{x}
\left(
2\beta_1\frac{y_{03}^3}{z_{03}}+z_{03} + 2 \alpha y_{03} z_{03}
\right)
.
\end{aligned}
$$
Base point:
$$
b_2\ :\ y_{03}=0,z_{03}=0.
$$
We note that $b_2$ is the intersection point of $\mathcal{L}_0$ and $\mathcal{L}_{\infty}$.

The flow of (\ref{eq:p3-system}) passes through $b_2$ at point $x$ if and only if $z$ has a pole at $x$ and $y$ either is regular at $x$ or has a pole of smaller order than $z$ at $x$.

\subsection{Resolution at $b_1$}\label{sec:b1-blow}
\paragraph{First chart:}
\begin{gather*}
y_{11}=\frac{y_{02}}{z_{02}}=\frac1{z},
\qquad
z_{11}=z_{02}=\frac{z}{y},
\\
y=\frac{1}{y_{11}z_{11}},
\qquad
z=\frac{1}{y_{11}}.
\end{gather*}
The exceptional line is $\mathcal{L}_1:z_{11}=0$.
The preimage of $\mathcal{L}_{\infty}$ is $y_{11}=0$, while $\mathcal{L}_0$ is not visible in this chart.

The Jacobian is:
$$
\begin{aligned}
J_{11}=&\frac{\partial y_{11}}{\partial y}\cdot\frac{\partial z_{11}}{\partial z}-\frac{\partial z_{11}}{\partial y}\cdot\frac{\partial y_{11}}{\partial z}
=-y_{11}^3 z_{11}^2,
\\
J'_{11}=&\frac{2 y_{11}^3 z_{11}^2 \left(\beta_1 y_{11}^3 z_{11}^2+\alpha  y_{11}-1\right)}{x}+\frac{1}{2} y_{11} z_{11} \left(4 \delta_1 y_{11}^3 z_{11}^2+4 \gamma  y_{11}^2+1\right).
\end{aligned}
$$

The vector field is:
$$
\begin{aligned}
y_{11}'&=\frac{1-4\gamma y_{11}^2}{2 y_{11} z_{11}}
-\frac{2 y_{11}^2 (\alpha + \beta_1 y_{11}^2 z_{11}^2 )}{x}
,
\\
z_{11}'&=2 \gamma - \frac{1}{y_{11}^2} - \delta_1 y_{11} z_{11}^2
+\frac{z_{11}(1 + 2\alpha y_{11}  + 2\beta_1 y_{11}^3 z_{11}^2 )}{x}
.
\end{aligned}
$$
Base points:
\begin{gather*}
b_3\ :\ y_{11}=\frac{1}{\sqrt{4\gamma}}, z_{11}=0,
\\
b_4\ :\ y_{11}=-\frac{1}{\sqrt{4\gamma}}, z_{11}=0.
\end{gather*}
The flow of (\ref{eq:p3-system}) passes through $b_3$ at point $x$ if and only if $y$ has a pole at $x$ and $z(x)=\sqrt{4\gamma}$.
The flow passes through $b_4$ if and only if $y$ has a pole and $z$ takes the value  $-\sqrt{4\gamma}$.

The energy is:
$$
\begin{aligned}
E=&\frac{1}{4 y_{11}^4 z_{11}^2}-\frac{\gamma }{y_{11}^2 z_{11}^2}+\frac{\delta_1}{y_{11}},
\\
E'=&\frac{4 \beta_1 \delta_1 y_{11}^6 z_{11}^4+4 \alpha  \delta_1 y_{11}^4 z_{11}^2+2 \beta_1 y_{11}^3 z_{11}^2+4 \gamma  y_{11}^2+2 \alpha  y_{11}-1}{2 x y_{11}^4 z_{11}^2}.
\end{aligned}
$$

\paragraph{Second chart:}
\begin{gather*}
y_{12}=y_{02}=\frac1{y},
\qquad
z_{12}=\frac{z_{02}}{y_{02}}=z,
\\
y=\frac{1}{y_{12}},
\qquad
z=z_{12}.
\end{gather*}
The exceptional line is $\mathcal{L}_1:y_{12}=0$.
The preimages of $\mathcal{L}_{\infty}$ and $\mathcal{L}_0$ are not visible in this chart.

The Jacobian:
$$
\begin{aligned}
J_{12}=&\frac{\partial y_{12}}{\partial y}\cdot\frac{\partial z_{12}}{\partial z}-\frac{\partial z_{12}}{\partial y}\cdot\frac{\partial y_{12}}{\partial z}=-y_{12}^2,
\\
J_{12}'=&y_{12} \left(2 \delta_1 y_{12}^2+z_{12}\right)-\frac{2 y_{12}^2}{x}.
\end{aligned}
$$

The vector field is:
$$
\begin{aligned}
y_{12}'&=-\delta_1 y_{12}^2-\frac{z_{12}}{2}+\frac{y_{12}}{x},
\\
z_{12}'&=\frac{4 \gamma - z_{12}^2}{2 y_{12}}
+\frac{2 (\alpha + \beta_1 y_{12}^2 )}{x}.
\end{aligned}
$$
Base points visible in this chart are again: 
$$
b_3(y_{12}=0,z_{12}=\sqrt{4\gamma}),
\quad
b_4(y_{12}=0,z_{12}=-\sqrt{4\gamma}).
$$

The energy:
$$
\begin{aligned}
E=&-\frac{\gamma }{y_{12}^2}+\frac{z_{12}^2}{4 y_{12}^2}+\delta_1 z_{12},
\\
E'=&\frac{4 \left(\gamma +\delta_1 y_{12}^2 \left(\alpha +\beta_1 y_{12}^2\right)\right)+2 z_{12} \left(\alpha +\beta_1 y_{12}^2\right)-z_{12}^2}{2 x y_{12}^2}.
\end{aligned}
$$

\subsection{Resolution at $b_2$}\label{sec:b2-blow}
\paragraph{First chart:}
\begin{gather*}
y_{21}=\frac{y_{03}}{z_{03}}=\frac1y,
\qquad
z_{21}=z_{03}=\frac{y}{z},
\\
y=\frac{1}{y_{21}},
\qquad
z=\frac{1}{y_{21}z_{21}}.
\end{gather*}
The exceptional line is $\mathcal{L}_2:z_{21}=0$.
The preimage of $\mathcal{L}_0$ is not visible in this chart, while the preimage of $\mathcal{L}_{\infty}$ is given by $y_{03}=0$.

The Jacobian:
$$
\begin{aligned}
J_{21}=&
\frac{\partial y_{21}}{\partial y}\cdot\frac{\partial z_{21}}{\partial z}-\frac{\partial z_{21}}{\partial y}\cdot\frac{\partial y_{21}}{\partial z}
=
y_{21}^3z_{21}^2
,
\\
J_{21}'=&
-\frac{y_{21}^3 z_{21}^2 \left(4 \beta_1 y_{21}^3 z_{21}+4 \alpha  y_{21} z_{21}-1\right)}{x}-\frac{1}{2} y_{21} z_{21} \left(2 \delta_1 y_{21}^3 z_{21}+8 \gamma  y_{21}^2 z_{21}^2-1\right)
.
\end{aligned}
$$

The vector field is:
$$
\begin{aligned}
y_{21}'&=-\delta_1 y_{21}^2 -\frac{1}{2 y_{21} z_{21}}+\frac{y_{21}}{x},
\\
z_{21}'&=\frac{1}{y_{21}^2} - 2 \gamma z_{21}^2 + \delta_1 y_{21} z_{21}
-\frac{z_{21}}{x} (1 + 2\alpha y_{21} z_{21}  +     2 \beta_1 y_{21}^3 z_{21}).
\end{aligned}
$$
No base points.

The energy:
$$
\begin{aligned}
E=&
\frac{1}{4 y_{21}^4 z_{21}^2}-\frac{\gamma }{y_{21}^2}+\frac{\delta_1}{y_{21} z_{21}}
,
\\
E'=&
\frac{4 \beta_1 \delta_1 y_{21}^6 z_{21}^2+4 \alpha  \delta_1 y_{21}^4 z_{21}^2+2 \beta_1 y_{21}^3 z_{21}+4 \gamma  y_{21}^2 z_{21}^2+2 \alpha  y_{21} z_{21}-1}{2 x y_{21}^4 z_{21}^2}
.
\end{aligned}
$$

\paragraph{Second chart:}
\begin{gather*}
y_{22}=y_{03}=\frac1z,
\qquad
z_{22}=\frac{z_{03}}{y_{03}}=y,
\\
y=z_{22},
\qquad
z=\frac{1}{y_{22}}.
\end{gather*}
The exceptional line is $\mathcal{L}_2:y_{22}=0$. The preimage of $\mathcal{L}_0$ is given by $z_{22}=0$, while $\mathcal{L}_{\infty}$ is not visible in this chart.

The Jacobian:
$$
\begin{aligned}
J_{22}=&
\frac{\partial y_{22}}{\partial y}\cdot\frac{\partial z_{22}}{\partial z}-\frac{\partial z_{22}}{\partial y}\cdot\frac{\partial y_{22}}{\partial z}
=y_{22}^2
,
\\
J_{22}'=&
y_{22} z_{22} \left(1-4 \gamma  y_{22}^2\right)-\frac{4 y_{22}^3 \left(\beta_1+\alpha  z_{22}^2\right)}{x z_{22}^2}
.
\end{aligned}
$$

The vector field is:
$$
\begin{aligned}
y_{22}'&= \frac{z_{22}}{2} - 2 \gamma y_{22}^2 z_{22}
-\frac{2 y_{22}^2 ( \alpha z_{22}^2+ \beta_1)}{z_{22}^2 x}
,
\\
z_{22}'&=\delta_1+\frac{z_{22}^2}{2 y_{22}}-\frac{z_{22}}{x}.
\end{aligned}
$$
Base point:
$$
b_5\ :\ y_{22}=0,z_{22}=0.
$$
This point belongs to the preimage of $\mathcal{L}_0$.

The energy:
$$
\begin{aligned}
E=&\frac{z_{22}^2}{4 y_{22}^2}+\frac{\delta_1}{y_{22}}-\gamma  z_{22}^2
,
\\
E'=&
\frac{4 \beta_1 \delta_1 y_{22}^2+z_{22}^4 \left(4 \gamma  y_{22}^2+2 \alpha  y_{22}-1\right)+2 y_{22} z_{22}^2 (\beta_1+2 \alpha  \delta_1 y_{22})}{2 x y_{22}^2 z_{22}^2}
.
\end{aligned}
$$

\subsection{Resolution at $b_3$}\label{sec:b3-blow}
\paragraph{First chart:}
\begin{gather*}
y_{31}=\frac{y_{12}}{z_{12}-\sqrt{4\gamma}}=\frac{1}{y(z-\sqrt{4\gamma})}
,
\qquad
z_{31}=z_{12}-\sqrt{4\gamma}=z-\sqrt{4\gamma},
\\
y=\frac{1}{y_{31}z_{31}},
\qquad
z=z_{31}+\sqrt{4\gamma}.
\end{gather*}
The exceptional line is $\mathcal{L}_3:z_{31}=0$.
The preimage of $\mathcal{L}_1$ is $y_{31}=0$.

The Jacobian:
$$
\begin{aligned}
J_{31}=&
\frac{\partial y_{31}}{\partial y}\cdot\frac{\partial z_{31}}{\partial z}-\frac{\partial z_{31}}{\partial y}\cdot\frac{\partial y_{31}}{\partial z}
=
-y_{31}^2 z_{31},
\\
J_{31}'=&
\frac{2 \beta_1 y_{31}^4 z_{31}^2}{x}+\frac{2 \alpha  y_{31}^2}{x}-\frac{2 y_{31}^2 z_{31}}{x}+2 \delta_1 y_{31}^3 z_{31}^2+\frac{y_{31} z_{31}}{2}
.
\end{aligned}
$$

The vector field is:
$$
\begin{aligned}
y_{31}'=&
\frac{x\sqrt{\gamma}-2 \alpha  y_{31}}{x z_{31}}
-\frac{2 \beta_1 y_{31}^3 z_{31}}{x}+\frac{y_{31}}{x}-\delta_1 y_{31}^2 z_{31}
 ,
\\
z_{31}'=&
\frac{4 y_{31} \left(\alpha +\beta_1 y_{31}^2 z_{31}^2\right)-x \left(4 \sqrt{\gamma }+z_{31}\right)}{2 x y_{31}}
.
\end{aligned}
$$
Base point:
$$
b_6\ :\ y_{31}=\frac{x\sqrt{\gamma}}{2\alpha},\ z_{31}=0.
$$

The energy:
$$
\begin{aligned}
E=&
\frac{\left(2 \sqrt{\gamma }+z_{31}\right)^2}{4 y_{31}^2 z_{31}^2}-\frac{\gamma }{y_{31}^2 z_{31}^2}+\delta_1 \left(2 \sqrt{\gamma }+z_{31}\right)
,
\\
E'=&
\frac{2 \alpha  \delta_1}{x}+\frac{2 \beta_1 \sqrt{\gamma }}{x}+\frac{2 \alpha  \sqrt{\gamma }}{x y_{31}^2 z_{31}^2}+\frac{2 \beta_1 \delta_1 y_{31}^2 z_{31}^2}{x}+\frac{\alpha }{x y_{31}^2 z_{31}}-\frac{2 \sqrt{\gamma }}{x y_{31}^2 z_{31}}-\frac{1}{2 x y_{31}^2}+\frac{\beta_1 z_{31}}{x}
.
\end{aligned}
$$

\paragraph{Second chart:}
\begin{gather*}
y_{32}=y_{12}=\frac1y,
\qquad
z_{32}=\frac{z_{12}-\sqrt{4\gamma}}{y_{12}}=y(z-\sqrt{4\gamma}),
\\
y=\frac{1}{y_{32}},
\qquad
z=y_{32}z_{32}+\sqrt{4\gamma}.
\end{gather*}
The exceptional line is $\mathcal{L}_3:y_{32}=0$.
$\mathcal{L}_1$ is not visible in this chart.

The Jacobian:
$$
\begin{aligned}
J_{32}=&
\frac{\partial y_{32}}{\partial y}\cdot\frac{\partial z_{32}}{\partial z}-\frac{\partial z_{32}}{\partial y}\cdot\frac{\partial y_{32}}{\partial z}
=
-y_{32},
\\
J_{32}'=&
\sqrt{\gamma }-\frac{y_{32}}{x}+\delta_1 y_{32}^2+\frac{y_{32} z_{32}}{2}
.
\end{aligned}
$$

The vector field is:
$$
\begin{aligned}
y_{32}'&=
-\sqrt{\gamma }+\frac{y_{32}}{x}-\delta_1 y_{32}^2-\frac{y_{32} z_{32}}{2}
,
\\
z_{32}'&=
\frac{2 \alpha-\sqrt{\gamma }x z_{32} }{x y_{32}}
+\frac{2 \beta_1 y_{32}}{x}-\frac{z_{32}}{x}+\delta_1 y_{32} z_{32}
.
\end{aligned}
$$
The only base point on the exceptional line is again 
$b_6\left(y_{32}=0,z_{32}=\frac{2\alpha}{x\sqrt{\gamma}}\right)$.

The energy:
$$
\begin{aligned}
E=&
2 \sqrt{\gamma } \delta_1+\frac{\sqrt{\gamma } z_{32}}{y_{32}}+\delta_1 y_{32} z_{32}+\frac{z_{32}^2}{4}
,
\\
E'=&
\frac{2 \alpha  \delta_1}{x}+\frac{2 \beta_1 \sqrt{\gamma }}{x}+\frac{2 \alpha  \sqrt{\gamma }}{x y_{32}^2}+\frac{2 \beta_1 \delta_1 y_{32}^2}{x}+\frac{\alpha  z_{32}}{x y_{32}}+\frac{\beta_1 y_{32} z_{32}}{x}-\frac{2 \sqrt{\gamma } z_{32}}{x y_{32}}-\frac{z_{32}^2}{2 x}
.
\end{aligned}
$$

\subsection{Resolution at $b_4$}\label{sec:b4-blow}
We note that this is the same as the resolution at $b_3$, when $\sqrt{\gamma}$ is replaced by $-\sqrt{\gamma}$.
We get a base point $b_7$ on the exceptional line $\mathcal{L}_4$.
$b_7$ does not belong to the preimage of $\mathcal{L}_1$.

\subsection{Resolution at $b_5$}\label{sec:b5-blow}

\paragraph{First chart:}
\begin{gather*}
y_{51}=\frac{y_{22}}{z_{22}}=\frac{1}{yz},
\qquad
z_{51}=z_{22}=y,
\\
y=z_{51},
\qquad
z=\frac{1}{y_{51}z_{51}}.
\end{gather*}
The exceptional line is $\mathcal{L}_5:z_{51}=0$.
In this chart, line $\mathcal{L}_2$ is $y_{51}=0$, while $\mathcal{L}_0$ is not visible.

The Jacobian:
$$
\begin{aligned}
J_{51}=&
\frac{\partial y_{51}}{\partial y}\cdot\frac{\partial z_{51}}{\partial z}-\frac{\partial z_{51}}{\partial y}\cdot\frac{\partial y_{51}}{\partial z}
=
y_{51}^2 z_{51},
\\
J_{51}'=&
-\frac{y_{51}^2 \left(4 \beta_1 y_{51}+4 \alpha  y_{51} z_{51}^2-z_{51}\right)}{x}-\frac{1}{2} y_{51} \left(8 \gamma  y_{51}^2 z_{51}^3+2 \delta_1 y_{51}-z_{51}\right)
.
\end{aligned}
$$

The vector field is:
$$
\begin{aligned}
y_{51}'&=-\frac{\delta_1 y_{51}}{z_{51}} - 2 \gamma y_{51}^2 z_{51}^2
+\frac{y_{51}}{x} - \frac{2\alpha y_{51}^2 z_{51} }{x} - \frac{2 \beta_1 y_{51}^2}{z_{51} x}
,
\\
z_{51}'&=\frac{z_{51}}{2 y_{51}}+\delta_1-\frac{z_{51}}{x}.
\end{aligned}
$$
Base points:
$$
\begin{aligned}
&b_8\ :\ y_{51}=0,\  z_{51}=0,
\\
&
b_{11}\ :\ y_{51}=-\frac{x\delta_1}{2\beta_1},\ z_{51}=0. 
\end{aligned}
$$
The point $b_8$ is the intersection of $\mathcal{L}_5$ and the preimage of $\mathcal{L}_2$.

The energy:
$$
\begin{aligned}
E=&
\frac{1}{4 y_{51}^2}+\frac{\delta_1}{y_{51} z_{51}}-\gamma  z_{51}^2
,
\\
E'=&
\frac{4 \beta_1 \delta_1 y_{51}^2+4 \gamma  y_{51}^2 z_{51}^4+z_{51}^2 \left(4 \alpha  \delta_1 y_{51}^2-1\right)+2 \alpha  y_{51} z_{51}^3+2 \beta_1 y_{51} z_{51}}{2 x y_{51}^2 z_{51}^2}
.
\end{aligned}
$$

\paragraph{Second chart:}
\begin{gather*}
y_{52}=y_{22}=\frac1{z},
\qquad
z_{52}=\frac{z_{22}}{y_{22}}=yz,
\\
y=y_{52}z_{52},
\qquad
z=\frac{1}{y_{52}}.
\end{gather*}
The exceptional line is $\mathcal{L}_5:y_{52}=0$.
The line $\mathcal{L}_2$ is not visible in this chart, while the preimage of $\mathcal{L}_0$ is $z_{52}=0$.

The Jacobian:
$$
\begin{aligned}
J_{52}=&
\frac{\partial y_{52}}{\partial y}\cdot\frac{\partial z_{52}}{\partial z}-\frac{\partial z_{52}}{\partial y}\cdot\frac{\partial y_{52}}{\partial z}
=
y_{52}
.
\end{aligned}
$$

The vector field is:
$$
\begin{aligned}
y_{52}'&=\frac{y_{52} z_{52}}{2} - 2\gamma y_{52}^3 z_{52}
-\frac{2\alpha y_{52}^2 }{x} - \frac{2 \beta_1}{z_{52}^2 x}
,
\\
z_{52}'&=2\gamma y_{52}^2z_{52}^2+\frac{ \delta_1}{y_{52}}
+\frac{2 \beta_1}{y_{52} z_{52} x} -\frac{z_{52}}{x} + \frac{2\alpha y_{52} z_{52}}{x}
.
\end{aligned}
$$
The base point is again $b_{11}\left(y_{52}=0,z_{52}=-\dfrac{2\beta_1}{x\delta_1}\right)$.
As $x$ grows, the point $b_{11}$ approaches the intersection point of $\mathcal{L}_5$ and the preimage of  $\mathcal{L}_0$.
However, the limit point $(y_{52}=0,z_{52}=0)$ is not a base point of the limit system (\ref{eq:p3-system-auto}).

The energy:
$$
\begin{aligned}
E=&
-\gamma y_{52}^2 z_{52}^2+\frac{\delta_1}{y_{52}}+\frac{z_{52}^2}{4}
,
\\
E'=&
\frac{4 \beta_1 \delta_1+4 \gamma  y_{52}^4 z_{52}^4+2 \alpha  y_{52}^3 z_{52}^4-y_{52}^2 \left(z_{52}^4-4 \alpha  \delta_1 z_{52}^2\right)+2 \beta_1 y_{52} z_{52}^2}{2 x y_{52}^2 z_{52}^2}
.
\end{aligned}
$$

\subsection{Resolution at $b_6$}\label{sec:b6-blow}

\paragraph{First chart:}
$$
\begin{aligned}
y_{61}=&y_{32}\left(z_{32}-\frac{2\alpha}{x\sqrt{\gamma}}\right)^{-1}
=\frac1y\left(y(z-\sqrt{4\gamma})-\frac{2\alpha}{x\sqrt{\gamma}}\right)^{-1}
,
\\
z_{61}=&z_{32}-\frac{2\alpha}{x\sqrt{\gamma}}=y(z-\sqrt{4\gamma})-\frac{2\alpha}{x\sqrt{\gamma}}
,
\\
y=&\frac{1}{y_{61}z_{61}},
\\
z=&
y_{61}z_{61}\left(z_{61}+\frac{2\alpha}{x\sqrt{\gamma}}\right)+\sqrt{4\gamma}
.
\end{aligned}
$$
The exceptional line is $\mathcal{L}_6:z_{61}=0$.
The preimage of $\mathcal{L}_3$ is $y_{61}=0$.

The Jacobian:
$$
\begin{aligned}
J_{61}=&
\frac{\partial y_{61}}{\partial y}\cdot\frac{\partial z_{61}}{\partial z}-\frac{\partial z_{61}}{\partial y}\cdot\frac{\partial y_{61}}{\partial z}
=
-y_{61}
.
\end{aligned}
$$

The vector field is:
$$
\begin{aligned}
y_{61}'=&
-\frac{2 \alpha  \delta_1 y_{61}^2}{\sqrt{\gamma } x}-\frac{2 \beta_1 y_{61}^2}{x}-\frac{\alpha  y_{61}}{\sqrt{\gamma } x}+\frac{2 y_{61}}{x}-2 \delta_1 y_{61}^2 z_{61}-\frac{y_{61} z_{61}}{2}
,
\\
z_{61}'=&
\frac{2 \alpha  \delta_1 y_{61} z_{61}}{\sqrt{\gamma } x}+\frac{2 \beta_1 y_{61} z_{61}}{x}-\frac{z_{61}}{x}-\frac{\sqrt{\gamma }}{y_{61}}+\delta_1 y_{61} z_{61}^2
.
\end{aligned}
$$
No base points.

The energy:
$$
\begin{aligned}
E=&
\frac{\alpha ^2}{\gamma  x^2}+\frac{z_{61} (\alpha +2 \alpha  \delta_1 y_{61})}{\sqrt{\gamma } x}+\frac{2 \alpha }{x y_{61} z_{61}}+\frac{\sqrt{\gamma } (2 \delta_1 y_{61}+1)}{y_{61}}+z_{61}^2 \left(\delta_1 y_{61}+\frac{1}{4}\right)
,
\\
E'=&
-\frac{2 \alpha ^2}{\gamma  x^3}
+\frac{2 \alpha ^2}{\sqrt{\gamma } x^2 y_{61} z_{61}}
+\frac{2 \alpha  \beta_1 y_{61} z_{61}}{\sqrt{\gamma } x^2}
-\frac{4 \alpha }{x^2 y_{61} z_{61}}-\frac{2 \alpha  z_{61}}{\sqrt{\gamma } x^2}
+\frac{2 \alpha  \delta_1}{x}
+\frac{2 \beta_1 \sqrt{\gamma }}{x}
\\&
+\frac{2 \alpha  \sqrt{\gamma }}{x y_{61}^2 z_{61}^2}
+\frac{2 \beta_1 \delta_1 y_{61}^2 z_{61}^2}{x}+\frac{\alpha }{x y_{61}}
-\frac{2 \sqrt{\gamma }}{x y_{61}}
+\frac{\beta_1 y_{61} z_{61}^2}{x}
-\frac{z_{61}^2}{2 x}
,
\\
EJ_{61}=&-\sqrt{\gamma }
-\frac{2 \alpha }{x z_{61}}
-\frac{\alpha ^2 y_{61}}{\gamma  x^2}
-\frac{2 \alpha  \delta_1 y_{61}^2 z_{61}}{\sqrt{\gamma } x}
-\frac{\alpha  y_{61} z_{61}}{\sqrt{\gamma } x}
-\delta_1 y_{61}^2 z_{61}^2
-2 \sqrt{\gamma } \delta_1 y_{61}
-\frac{y_{61} z_{61}^2}{4}
.
\end{aligned}
$$

\paragraph{Second chart:}
$$
\begin{aligned}
y_{62}=&y_{32}=\frac1y,
\\
z_{62}=&
\frac{z_{32}-\frac{2\alpha}{x\sqrt{\gamma}}}{y_{32}}
=y\left(y(z-\sqrt{4\gamma})-\frac{2\alpha}{x\sqrt{\gamma}}\right)
\\
y=&\frac1{y_{62}}
,
\\
z=&y_{62}
\left(
y_{62}z_{62}+\frac{2\alpha}{x\sqrt{\gamma}}
\right)
+\sqrt{4\gamma}
.
\end{aligned}
$$
The exceptional line is $\mathcal{L}_6:y_{62}=0$.
$\mathcal{L}_3$ is not visible in this chart.

The Jacobian:
$$
\begin{aligned}
J_{62}=&
\frac{\partial y_{62}}{\partial y}\cdot\frac{\partial z_{62}}{\partial z}-\frac{\partial z_{62}}{\partial y}\cdot\frac{\partial y_{62}}{\partial z}
=
-1
.
\end{aligned}
$$

The vector field is:
$$
\begin{aligned}
y_{62}'=&
-\sqrt{\gamma }-\frac{\alpha  y_{62}}{\sqrt{\gamma } x}+\frac{y_{62}}{x}-\delta_1 y_{62}^2-\frac{y_{62}^2 z_{62}}{2}
,
\\
z_{62}'=&
\frac{2 \alpha  \delta_1}{\sqrt{\gamma } x}+\frac{2 \beta_1}{x}+\frac{\alpha  z_{62}}{\sqrt{\gamma } x}-\frac{2 z_{62}}{x}+\frac{y_{62} z_{62}^2}{2}+2 \delta_1 y_{62} z_{62}
.
\end{aligned}
$$
No base points.

The energy:
$$
\begin{aligned}
E=&
2 \sqrt{\gamma } \delta_1+\frac{\alpha ^2}{\gamma  x^2}+\frac{2 \alpha  \delta_1 y_{62}}{\sqrt{\gamma } x}+\frac{2 \alpha }{x y_{62}}+\frac{\alpha  y_{62} z_{62}}{\sqrt{\gamma } x}+\frac{y_{62}^2 z_{62}^2}{4}+\delta_1 y_{62}^2 z_{62}+\sqrt{\gamma } z_{62}
.
\end{aligned}
$$

\subsection{Resolution at $b_7$}\label{sec:b7-blow}

This resolution is equivalent to the one at $b_6$, after replacing $\sqrt{\gamma}$ by $-\sqrt{\gamma}$.
We denote by $\mathcal{L}_7$ the exceptional line obtained by blowing up $b_7$.

No new base points occur after this resolution.

\subsection{Resolution at $b_8$}\label{sec:b8-blow}

\paragraph{First chart:}
\begin{gather*}
y_{81}=\frac{y_{51}}{z_{51}}=\frac{1}{y^2z},
\qquad
z_{81}=z_{51}=y,
\\
y=z_{81},
\qquad
z=\frac{1}{y_{81}z_{81}^2}.
\end{gather*}
The exceptional line is $\mathcal{L}_8:z_{81}=0$.
The preimage of $\mathcal{L}_2$ is $y_{81}=0$, while $\mathcal{L}_5$ is not visible in this chart.

The vector field is:
$$
\begin{aligned}
y_{81}'&=-\frac{1 + 4 \delta_1 y_{81}}{2z_{81}} - 2 \gamma y_{81}^2 z_{81}^3
-\frac{2 y_{81} (\alpha y_{81} z_{81}^2  + \beta_1 y_{81}-1)}{x}
,
\\
z_{81}'&=\delta_1 + \frac{1}{2 y_{81}}-\frac{z_{81}}{x}.
\end{aligned}
$$
Base point:
$$
b_9\ :\ y_{81}=-\frac{1}{4\delta_1}, z_{81}=0.
$$

The energy:
$$
\begin{aligned}
E=&
\frac{1}{4 y_{81}^2 z_{81}^2}+\frac{\delta_1}{y_{81} z_{81}^2}-\gamma  z_{81}^2
,
\\
E'=&
\frac{4 y_{81}^2 \left(\beta_1 \delta_1+\gamma  z_{81}^4+\alpha  \delta_1 z_{81}^2\right)+2 y_{81} \left(\beta_1+\alpha  z_{81}^2\right)-1}{2 x y_{81}^2 z_{81}^2}
.
\end{aligned}
$$

\paragraph{Second chart:}
\begin{gather*}
y_{82}=y_{51}=\frac{1}{yz},
\qquad
z_{82}=\frac{z_{51}}{y_{51}}=y^2z,
\\
y=y_{82}z_{82},
\qquad
z=\frac{1}{y_{82}^2z_{82}}.
\end{gather*}
The exceptional line is $\mathcal{L}_8:y_{82}=0$.
The preimage of $\mathcal{L}_5$ is $z_{82}=0$, while $\mathcal{L}_2$ is not visible in this chart.

The Jacobian:
$$
\begin{aligned}
J_{82}=&
\frac{\partial y_{82}}{\partial y}\cdot\frac{\partial z_{82}}{\partial z}-\frac{\partial z_{82}}{\partial y}\cdot\frac{\partial y_{82}}{\partial z}
=
y_{82}^2 z_{82},
\\
J_{82}'=&
-\frac{2 \alpha  y_{82}^4 z_{82}^2}{x}-\frac{2 \beta_1 y_{82}^2}{x}-2 \gamma  y_{82}^5 z_{82}^3+\frac{y_{82} z_{82}}{2},
\\
\frac{J_{82}'}{J_{82}}=&
-\frac{2 \alpha  y_{82}^2 z_{82}}{x}-\frac{2 \beta_1}{x z_{82}}-2 \gamma  y_{82}^3 z_{82}^2+\frac{1}{2 y_{82}}.
\end{aligned}
$$

The vector field is:
$$
\begin{aligned}
y_{82}'&=-\frac{\delta_1x+2\beta_1 y_{82}}{z_{82}} - 2 \gamma y_{82}^4 z_{82}^2
\frac{y_{82}}{x} - \frac{2\alpha y_{82}^3 z_{82}}{x} 
,
\\
z_{82}'&=\frac{4 \delta_1 + z_{82}}{2y_{82}} + 2 \gamma y_{82}^3 z_{82}^3
-\frac{2 z_{82}}{x} + \frac{2\alpha y_{82}^2 z_{82}^2}{x} + \frac{2 \beta}{x}
.
\end{aligned}
$$
The only base point in this chart is again $b_9(y_{82}=0,z_{82}=-4\delta_1)$.

The energy:
$$
\begin{aligned}
E=&-\gamma y_{82}^2 z_{82}^2+\frac{\delta_1}{y_{82}^2 z_{82}}+\frac{1}{4 y_{82}^2}
,
\\
E'=&
\frac{4 \beta_1 \delta_1+4 \gamma  y_{82}^4 z_{82}^4+2 \alpha  y_{82}^2 z_{82}^3+z_{82}^2 \left(4 \alpha  \delta_1 y_{82}^2-1\right)+2 \beta_1 z_{82}}{2 x y_{82}^2 z_{82}^2}
,
\\
\frac{E'}{E}=&
\frac{2 \left(4 \beta_1 \delta_1+4 \gamma  y_{82}^4 z_{82}^4+2 \alpha  y_{82}^2 z_{82}^3+z_{82}^2 \left(4 \alpha  \delta_1 y_{82}^2-1\right)+2 \beta_1 z_{82}\right)}{x z_{82} \left(4 \delta_1-4 \gamma  y_{82}^4 z_{82}^3+z_{82}\right)}
,
\\
EJ_{82}=&
\delta_1-\gamma  y_{82}^4 z_{82}^3+\frac{z_{82}}{4}.
\end{aligned}
$$

\subsection{Resolution at $b_9$}\label{sec:b9-blow}

\paragraph{First chart:}
\begin{gather*}
y_{91}=\frac{y_{81}+\frac{1}{4\delta_1}}{z_{81}}=\frac1y\left(\frac{1}{y^2z}+\frac{1}{4\delta_1}\right),
\qquad
z_{91}=z_{81}=y,
\\
y=z_{91},
\qquad
z=\frac{4 \delta_1 }{z_{91}^2 (4 \delta_1  y_{91} z_{91}-1)}.
\end{gather*}
The exceptional line is $\mathcal{L}_9:z_{91}=0$.
The line $\mathcal{L}_8$ is not visible in this chart.

The vector field is:
$$
\begin{aligned}
y_{91}'=&-\frac{\gamma}{8\delta_1^2}  z_{91}^2+ \frac{\gamma}{\delta_1}  y_{91} z_{91}^3-2 \gamma y_{91}^2 z_{91}^4
+\frac{\delta_1 y_{91} (1 - 12\delta_1 y_{91} z_{91})}{z_{91} ( 4\delta_1 y_{91} z_{91}-1)}
\\&
+\frac{(3\delta_1+\beta_1) y_{91}}{x} - \frac{2 y_{91}^2 z_{91}^3 \alpha}{x} 
-  z_{91} \left(\frac{2 y_{91}^2 \beta_1}{x} +\frac{\alpha}{8 x \delta_1^2}\right) 
+ \frac{y_{91} z_{91}^2 \alpha}{x \delta_1}
\\
&-\frac{\beta_1 + 4 \delta_1}{8 z_{91} x \delta_1^2}
\\
z_{91}'&=\delta_1+\frac{2\delta_1}{ 4 \delta_1 y_{91} z_{91}-1}
-\frac{z_{91}}{x}.
\end{aligned}
$$
Base point:
$$
b_{10}\ :\ y_{91}=-\frac{\beta_1+4\delta_1}{8x\delta_1^3},\ z_{91}=0.
$$

\paragraph{Second chart:}
$$
\begin{aligned}
y_{92}=&y_{81}+\frac{1}{4\delta_1}=\frac{1}{y^2z}+\frac{1}{4\delta_1},
\\
z_{92}=&\frac{z_{81}}{y_{81}+\frac{1}{4\delta_1}}=\frac{y}{\frac{1}{y^2z}+\frac{1}{4\delta_1}},
\\
y=&y_{92}z_{92},
\\
z=&\frac{4\delta_1}{y_{92}^2z_{92}^2(4\delta_1 y_{92}-1)}.
\end{aligned}
$$
The exceptional line is $\mathcal{L}_9:y_{92}=0$.
The preimage of line $\mathcal{L}_8$ is $z_{92}=0$, and the preimage of $\mathcal{L}_2$ is $y_{92}=\frac{1}{4\delta_1}$.

The Jacobian:
$$
\begin{aligned}
J_{92}=&
\frac{\partial y_{92}}{\partial y}\cdot\frac{\partial z_{92}}{\partial z}-\frac{\partial z_{92}}{\partial y}\cdot\frac{\partial y_{92}}{\partial z}
=\frac{y_{92} z_{92}^2 (1-4 \delta_1 y_{92})^2}{16 \delta_1^2},
\\
J_{92}'=&
\frac{z_{92}^2(1-4 \delta_1 y_{92} )^2}{128 \delta_1^4 x} \times
\\
&\qquad\times
 \left(\beta_1+4 \delta_1
 -16\delta_1^2 y_{92}^4 z_{92}^2( \gamma   x y_{92} z_{92}+\alpha  )
 +\gamma  x y_{92}^3 z_{92}^3
 -y_{92}^2 \left(16 \beta_1 \delta_1^2-\alpha  z_{92}^2\right)\right),
\\
\frac{J_{92}'}{J_{92}}=&
-\frac{2 \alpha  y_{92}^3 z_{92}^2}{x}+\frac{\beta_1}{8 \delta_1^2 x y_{92}}-\frac{2 \beta_1 y_{92}}{x}+\frac{1}{2 \delta_1 x y_{92}}+\frac{\alpha  y_{92} z_{92}^2}{8 \delta_1^2 x}-2 \gamma  y_{92}^4 z_{92}^3+\frac{\gamma  y_{92}^2 z_{92}^3}{8 \delta_1^2}
.
\end{aligned}
$$

The vector field is:
$$
\begin{aligned}
y_{92}'=&-\frac{2\delta_1}{z_{92}} - \frac{\gamma}{8\delta_1^2} y_{92}^3 z_{92}^3 +
 \frac{\gamma}{\delta_1}  y_{92}^4 z_{92}^3- 2\gamma y_{92}^5 z_{92}^3
\\& 
+\frac{(2\delta_1+\beta) y_{92}}{\delta_1x} - 
\frac{2 y_{92}^4 z_{92}^2 \alpha}{x} 
- \frac{2 y_{92}^2 \beta_1}{x}
- \frac{ y_{92}^2 z_{92}^2 \alpha}{8 x \delta_1^2}
- \frac{\beta_1+4\delta_1}{ 8 x \delta_1^2} 
+ \frac{y_{92}^3 z_{92}^2 \alpha}{ x \delta_1},
\\
z_{92}'=&\frac{\gamma}{8\delta_1^2}y_{92}^2 z_{92}^4- \frac{\gamma}{\delta_1} y_{92}^3 z_{92}^4 + 
 2\gamma y_{92}^4 z_{92}^4
\\& 
-\frac{(3\delta_1+\beta_1) z_{92}}{\delta_1x} + \frac{2 y_{92}^3 z_{92}^3 \alpha}{x}
+  y_{92} \left(\frac{2 z_{92} \beta_1}{x} + \frac{z_{92}^3 \alpha}{8 x \delta_1^2}\right)
- \frac{y_{92}^2 z_{92}^3 \alpha}{x \delta_1} 
\\
&
-\frac{z_{92} \beta_1 + 4 z_{92} \delta_1 + 24 x \delta_1^3}{2 x \delta_1(1 - 4 y_{92} \delta_1)}
-\frac{z_{92} \beta_1 + 4 z_{92} \delta_1 + 8 x \delta_1^3}{8 y_{92} x \delta_1^2 ( 4 y_{92} \delta_1-1)}.
\end{aligned}
$$
The base point is again
$
b_{10}\left(y_{92}=0,z_{92}=-\dfrac{8x\delta_1^3}{\beta_1+4\delta_1}\right).
$

The energy:
$$
\begin{aligned}
E=&\frac{16 \delta_1^3-16 \gamma  \delta_1^2 y_{92}^5 z_{92}^4+8 \gamma  \delta_1 y_{92}^4 z_{92}^4-\gamma  y_{92}^3 z_{92}^4}{y_{92} z_{92}^2 (1-4 \delta_1 y_{92})^2},
\\
E'=&
\frac{2}{x y_{92}^2 z_{92}^2 (1-4 \delta_1 y_{92})^2}
\times
\\&
\quad
\times
  \left(-\delta_1 (\beta_1+4 \delta_1)+16 \gamma  \delta_1^2 y_{92}^6 z_{92}^4-8 \gamma  \delta_1 y_{92}^5 z_{92}^4
  +y_{92}^4 \left(\gamma  z_{92}^4+16 \alpha  \delta_1^3 z_{92}^2\right)
  \right.
  \\&
  \qquad\qquad\left.
  +y_{92}^2 \left(16 \beta_1 \delta_1^3-\alpha  \delta_1 z_{92}^2\right)\right)
,
\\
\frac{E'}{E}=&
\frac2{x y_{92} \left(-16 \delta_1^3+16 \gamma  \delta_1^2 y_{92}^5 z_{92}^4-8 \gamma  \delta_1 y_{92}^4 z_{92}^4+\gamma  y_{92}^3 z_{92}^4\right)}
\times
\\
&\quad\times
\left(\delta_1 (\beta_1+4 \delta_1)-16 \gamma  \delta_1^2 y_{92}^6 z_{92}^4+8 \gamma  \delta_1 y_{92}^5 z_{92}^4
-y_{92}^4 \left(\gamma  z_{92}^4+16 \alpha  \delta_1^3 z_{92}^2\right)
\right.
\\&\qquad\qquad
\left.
+y_{92}^2 \left(\alpha  \delta_1 z_{92}^2-16 \beta_1 \delta_1^3\right)\right)
.
\end{aligned}
$$

\subsection{Resolution at $b_{10}$}\label{sec:b10-blow}

\paragraph{First chart:}
$$
\begin{aligned}
y_{101}=&\frac1{z_{91}}\left(y_{91}+\frac{\beta_1+4\delta_1}{8x\delta_1^3}\right)
=\frac1{y^2}\left(\frac{1}{y^2z}+\frac{1}{4\delta_1}\right)+\frac{1}{y}\cdot\frac{\beta_1+4\delta_1}{8x\delta_1^3},
\\
z_{101}=&z_{91}=y,
\\
y=&z_{101},
\\
z=&\frac{8 \delta_1 ^3 x}{z_{101}^2 \left(8 \delta_1 ^3 y_{101} z_{101}^2 x-(\beta_1  +4 \delta_1)  z_{101}-2 \delta_1 ^2 x\right)}.
\end{aligned}
$$
The exceptional line is $\mathcal{L}_{10}:z_{101}=0$.
The line $\mathcal{L}_9$ is not visible in this chart.

The vector field is:
$$
\begin{aligned}
y_{101}'=&
-\frac{(4A\delta_1+1)^2}{4A\delta_1z_{101}^3} 
+ \frac{(1 + 4 A \delta_1) }{16 A z_{101}^2 x \delta_1^3}
\left(\beta_1 + 4 \delta_1 - 2 A\delta_1( \beta_1  + 8 \delta_1) + 8 A^2 \beta_1 \delta_1^2\right)
\\&
- \frac{2 A^2 \alpha}{x}
- \frac{ 2 A^2 \beta_1}{z_{101}^2 x}
- 2 A^2 z_{101} \gamma 
,
\\
z_{101}'=&\frac{1}{2 A} - \frac{z_{101}}{x} + \delta_1,
\end{aligned}
$$
with
$$
A=-\frac{1}{4 \delta_1} + 
 z_{101} \left(y_{101}z_{101} - \frac{\beta_1 + 4 \delta_1}{8  x \delta_1^3}\right).
$$
No base points.

\paragraph{Second chart:}
$$
\begin{aligned}
y_{102}=&y_{91}+\frac{\beta_1+4\delta_1}{8x\delta_1^3}
=\frac1y\left(\frac{1}{y^2z}+\frac{1}{4\delta_1}\right)+\frac{\beta_1+4\delta_1}{8x\delta_1^3},
\\
z_{102}=&z_{91}\left(y_{91}+\frac{\beta_1+4\delta_1}{8x\delta_1^3}\right)^{-1}
=y\left(\frac1y\left(\frac{1}{y^2z}+\frac{1}{4\delta_1}\right)+\frac{\beta_1+4\delta_1}{8x\delta_1^3}\right)^{-1},
\\
y=&y_{102}z_{102},
\\
z=&\frac{8 \delta_1 ^3 x}{y_{102}^2 z_{102}^2 \left(8 \delta_1 ^3 y_{102}^2 z_{102} x-\beta_1  y_{102} z_{102}-4 \delta_1  y_{102} z_{102}-2 \delta_1 ^2 x\right)}.
\end{aligned}
$$
The exceptional line is $\mathcal{L}_{10}:y_{102}=0$.
The preimage of $\mathcal{L}_9$ is $z_{102}=0$.

The Jacobian:
$$
\begin{aligned}
J_{102}=&
\frac{\partial y_{102}}{\partial y}\cdot\frac{\partial z_{102}}{\partial z}-\frac{\partial z_{102}}{\partial y}\cdot\frac{\partial y_{102}}{\partial z}
=
\frac{z_{102} \left(2 \delta_1^2 x-8 \delta_1^3 x y_{102}^2 z_{102}+y_{102} z_{102} (\beta_1+4 \delta_1)\right)^2}{64 \delta_1^6 x^2}.
\end{aligned}
$$

The vector field is:
$$
\begin{aligned}
y_{102}'=&-\frac{(1 + A \delta_1) (2 + 3 A \delta_1)}{4 A y_{102}^2 z_{102}^2 \delta_1},
\\
z_{102}'=&-\frac{(1 + A \delta_1) (2 + 3 A \delta_1)}{4 A y_{102}^2 z_{102}^2 \delta_1}
+\frac{2 + 6 A \delta_1 - A^2 \beta_1 \delta_1}{8 y_{102} z_{102} x \delta_1}
\\&
-\frac{1}{8} A^2 y_{102}^2 z_{102}^2 \gamma - \frac{\beta_1}{8 x^2 \delta_1^3}
- \frac{1}{ 2 x^2 \delta_1^2}
 ,
\end{aligned}
$$
with
$$
A=-\frac{1}{\delta_1} +  4 y_{102} z_{102} \left(y_{102} - \frac{\beta_1 + 4 \delta_1}{8 x \delta_1^3}\right).
$$
No base points.

The energy:
$$
\begin{aligned}
E=&-\gamma  y_{102}^2 z_{102}^2
+
\frac{8 \delta_1^4 x \left(-\beta_1-4 \delta_1+8 \delta_1^3 x y_{102}\right)}{y_{102} z_{102} \left(2 \delta_1^2 x-8 \delta_1^3 x y_{102}^2 z_{102}+y_{102} z_{102} (\beta_1+4 \delta_1)\right)^2},
\\
EJ_{102}=&-\frac{\beta_1+4 \delta_1-8 \delta_1^3 x y_{102}}{8 \delta_1^2 x y_{102}}
-
\gamma y_{102}^2 z_{102}^3 \frac{ \left(2 \delta_1^2 x-8 \delta_1^3 x y_{102}^2 z_{102}+y_{102} z_{102} (\beta_1+4 \delta_1)\right)^2}{64 \delta_1^6 x^2}.
\end{aligned}
$$

\subsection{Resolution at $b_{11}$}\label{sec:b11-blow}

\paragraph{First chart:}
$$
\begin{aligned}
y_{111}=&y_{52}\left(z_{52}+\frac{2\beta_1}{x\delta_1}\right)^{-1}
=\frac1z\left(yz+\frac{2\beta_1}{x\delta_1}\right)^{-1}
,
\\
z_{111}=&z_{52}+\frac{2\beta_1}{x\delta_1}=yz+\frac{2\beta_1}{x\delta_1},
\\
y=&\frac{y_{111} z_{111} (\delta_1 x  z_{111} - 2 \beta_1)}{\delta_1  x},
\\
z=&\frac{1}{y_{111}z_{111}}.
\end{aligned}
$$
The exceptional line is $\mathcal{L}_{11}:z_{111}=0$.
The preimage of $\mathcal{L}_5$ is $y_{111}=0$ and of $\mathcal{L}_0$ : $z_{111}=\frac{2\beta_1}{x\delta_1}$.

The Jacobian:
$$
J_{111}=y_{111}.
$$

The vector field is:
$$
\begin{aligned}
y_{111}'=&
\frac{y_{111}}{x} - 4 y_{111}^3 z_{111}^3 \gamma + 
 z_{111} \left(\frac{y_{111}}{2} - \frac{4 y_{111}^2 \alpha}{x} - \frac{8 y_{111}^3 \beta_1^2 \gamma}{ x^2 \delta_1^2}\right) -\frac{y_{111} \beta_1}{x \delta_1}
\\&
 + \frac{ 4 y_{111}^2 \alpha \beta_1}{x^2 \delta_1}
 + \frac{ 12 y_{111}^3 z_{111}^2 \beta_1 \gamma}{x \delta_1}
-\frac{x^2 \delta_1^3}{(-2 \beta_1 + z_{111} x \delta_1)^2}
 ,
\\
z_{111}'=&
2 y_{111}^2 z_{111}^4 \gamma + 
 z_{111}^2 \left(\frac{2 y_{111} \alpha}{x} + \frac{8 y_{111}^2 \beta_1^2 \gamma}{x^2 \delta_1^2}\right) 
 -
 z_{111} \left(\frac1x + \frac{4 y_{111} \alpha \beta_1}{x^2 \delta_1}\right)
\\&
 - \frac{8 y_{111}^2 z_{111}^3 \beta_1 \gamma}{x \delta_1}
 -\frac{x \delta_1^2}{y_{111} (2 \beta_1 - z_{111} x \delta_1)}
.
\end{aligned}
$$
No base points.

The energy:
$$
E=\frac{4 \delta_1^3 x^2-4 \gamma  y_{111}^3 z_{111}^3 (\delta_1 x z_{111}-2 \beta_1)^2+y_{111} z_{111} (\delta_1 x z_{111}-2 \beta_1)^2}{4 \delta_1^2 x^2 y_{111} z_{111}}.
$$

\paragraph{Second chart:}
$$
\begin{aligned}
y_{112}=&y_{52}=\frac1z,
\\
z_{112}=&\frac{1}{y_{52}}\left(z_{52}+\frac{2\beta_1}{x\delta_1}\right)
=z\left(yz+\frac{2\beta_1}{x\delta_1}\right),
\\
y=&y_{112}\left(y_{112}z_{112}-\frac{2\beta_1}{\delta_1 x}\right),
\\
z=&\frac{1}{y_{112}}.
\end{aligned}
$$
The exceptional line is $\mathcal{L}_{11}:y_{112}=0$.
The line $\mathcal{L}_5$ is not visible in this chart.

The vector field is:
$$
\begin{aligned}
y_{112}'=&
y_{112}^2 \left(\frac{z_{112}}{2} - \frac{2 \alpha}{x}\right) 
- 2 y_{112}^4 z_{112} \gamma - \frac{ y_{112} \beta_1}{x \delta_1} 
+ \frac{4 y_{112}^3 \beta_1 \gamma}{ x \delta_1}
- \frac{ 2 x \beta_1 \delta_1^2}{(2 \beta_1 - y_{112} z_{112} x \delta_1)^2}
,
\\
z_{112}'=&
z_{112}^2 \left( 4 y_{112}^3 \gamma-\frac{y_{112}}{2} \right) 
\\&
+ 
 z_{112} 
 \left(
 -\frac{1}{x} + \frac{4 y_{112} \alpha}{x} - \frac{\beta_1}{x^2 \delta_1} 
 + \frac{2 \beta_1}{x \delta_1} + \frac{4 y_{112}^2 \beta_1 \gamma}{x^2 \delta_1} 
 - \frac{16 y_{112}^2 \beta_1 \gamma}{x\delta_1}\right) 
\\& 
 + \frac{ 2 \beta_1^2}{y_{112} x^3 \delta_1^2} - \frac{2 \beta_1^2}{ y_{112} x^2 \delta_1^2} 
 - \frac{8 y_{112} \beta_1^2 \gamma}{x^3 \delta_1^2}
 + \frac{ 16 y_{112} \beta_1^2 \gamma}{x^2 \delta_1^2}
 + \frac{4 \alpha \beta_1}{ x^3 \delta_1} 
 - \frac{8 \alpha \beta_1}{x^2 \delta_1}
\\& 
 + \frac{\delta_1}{y_{112}^2}
+ \frac{4 \beta_1 \delta_1 (\beta_1 - 2 x \beta_1 +    y_{112} z_{112} x^2 \delta_1)}
   {y_{112}^2 x (-2 \beta_1 + y_{112} z_{112} x \delta_1)^2} 
.
\end{aligned}
$$
No base points.

\subsection{Blowing down the preimage of $\mathcal{L}_0$}\label{sec:blowdownL0}

$$
\begin{aligned}
Y_{111}&=y_{111}\left(z_{111}-\frac{2\beta_1}{x\delta_1}\right)=y\left(yz+\frac{2\beta_1}{x\delta_1}\right)^{-1},
\\
Z_{111}&=z_{111}=yz+\frac{2\beta_1}{x\delta_1},
\\
y&=Y_{111}Z_{111},
\\
z&=\frac{Z_{111}-\frac{2\beta_1}{x\delta_1}}{Y_{111}Z_{111}}.
\end{aligned}
$$

The projection of $\mathcal{L}_{11}$ is $Z_{111}=0$ and of $\mathcal{L}_5$: $Y_{111}=0$.

The Jacobian and energy:
$$
\begin{aligned}
J_{111}=&
\frac{\partial Y_{111}}{\partial y}\cdot\frac{\partial Z_{111}}{\partial z}-\frac{\partial Z_{111}}{\partial y}\cdot\frac{\partial Y_{111}}{\partial z}
=
Y_{111},
\\
E=&
\frac{\beta_1^2}{\delta_1^2 x^2}-\frac{2 \beta_1}{x Y_{111} Z_{111}}-\frac{\beta_1 Z_{111}}{\delta_1 x}
-\gamma  Y_{111}^2 Z_{111}^2+\frac{\delta_1}{Y_{111}}+\frac{Z_{111}^2}{4}.
\end{aligned}
$$

The vector field:
$$
\begin{aligned}
Y_{111}'=&
-\frac{2 \alpha  Y_{111}^2}{x}-\frac{\beta_1 Y_{111}}{\delta_1 x}-2 \gamma  Y_{111}^3 Z_{111}+\frac{Y_{111} Z_{111}}{2}
,\\
Z_{111}'=&
\frac{2 \alpha  Y_{111} Z_{111}}{x}-\frac{Z_{111}}{x}+2 \gamma  Y_{111}^2 Z_{111}^2+\frac{\delta_1}{Y_{111}}
.
\end{aligned}
$$

%% file: B-notation/notation.tex
In this appendix, we collect the notation for base points, provide the charts in which they are defined and their coordinates in these charts, and give the relationships between constants used in the paper. 

\renewcommand{\arraystretch}{2}
\begin{longtable}[h]{|c|l|c|}
\hline
base point & coordinate system & coordinates\\
\hline
\endhead
\hline
\endfoot
  $b_1$ & $(y_{02},z_{02})=(\dfrac1y,\dfrac{z}{y})$ & $(0,0)$ 
  \\
  $b_2$ & $(y_{03},z_{03})=(\dfrac1{z},\dfrac{y}{z})$ & $(0,0)$  
  \\
  $b_3$ & $(y_{12},z_{12})=(y_{02},\dfrac{z_{02}}{y_{02}})=(\dfrac1y,z)$ & $(0,\sqrt{4\gamma})$
  \\
  $b_4$ & $(y_{12},z_{12})=(y_{02},\dfrac{z_{02}}{y_{02}})=(\dfrac1y,z)$ & $(0,-\sqrt{4\gamma})$
  \\  
  $b_5$ & $(y_{22},z_{22})=(y_{03},\dfrac{z_{03}}{y_{03}})=(\dfrac1z,y)$ & $(0,0)$
  \\  
  $b_6$ & $(y_{32},z_{32})=(y_{12},\dfrac{z_{12}-\sqrt{4\gamma}}{y_{12}})=(\dfrac1{y},y(z-\sqrt{4\gamma}))$ 
  & $(0,\dfrac{2\alpha}{x\sqrt{\gamma}})$
  \\
   $b_7$ & $(y_{42},z_{42})=(y_{12},\dfrac{z_{12}+\sqrt{4\gamma}}{y_{12}})=(\dfrac1{y},y(z+\sqrt{4\gamma}))$ 
  & $(0,-\dfrac{2\alpha}{x\sqrt{\gamma}})$
  \\
   $b_8$ & $(y_{51},z_{51})=(\dfrac{y_{22}}{z_{22}},z_{22})=(\dfrac1{yz},y)$ & $(0,0)$
  \\  
   $b_9$ & $(y_{81},z_{81})=(\dfrac{y_{51}}{z_{51}},z_{51})=(\dfrac1{y^2z},y)$ & $(-\dfrac1{4\delta_1},0)$
  \\   
   $b_{10}$ & 
   $(y_{91},z_{91})=(\dfrac{1}{z_{81}}(y_{81}+\dfrac{1}{4\delta_1}),z_{81})=(\dfrac1y(\dfrac1{y^2z}+\dfrac1{4\delta_1}),y)$ & $(-\dfrac{\beta_1+4\delta_1}{8x\delta_1^3},0)$
  \\   
   $b_{11}$ & $(y_{52},z_{52})=(y_{22},\dfrac{z_{22}}{y_{22}})=(\dfrac1z,yz)$ & $(0,-\dfrac{2\beta_1}{\delta_1 x})$
\end{longtable}

Note that the constants used throughout the paper are related as follows.
\begin{align*}
&\delta_1=i\sqrt{\delta},
\\
&\beta_1=\beta-2\delta_1.
\end{align*}